\documentclass[aps,pra,twocolumn,superscriptaddress,floatfix,nofootinbib,showpacs,longbibliography]{revtex4-2}

\usepackage[utf8]{inputenc}
\usepackage[T1]{fontenc}     
\usepackage[british]{babel}  
\usepackage[sc,osf]{mathpazo}
\usepackage{times}
\usepackage[table]{xcolor}
\usepackage[scaled=0.86]{berasans}  
\usepackage[colorlinks=true, citecolor=blue, urlcolor=blue]{hyperref}
\usepackage{comment}
\makeatletter
\newcommand{\setword}[2]{%
  \phantomsection
  #1\def\@currentlabel{\unexpanded{#1}}\label{#2}%
}
\makeatother
\usepackage{comment}
\usepackage{graphicx} 
\usepackage[babel]{microtype}  
\usepackage{amsmath,amssymb,amsthm,bm,amsfonts,mathrsfs,bbm} 

\usepackage{xspace}  
\usepackage{pgf,tikz}
\usepackage{xcolor}
\usepackage{multirow}
\usepackage{array}
\usepackage{bigstrut}
\usepackage{braket}
\usepackage{color}
\usepackage{natbib}
\usepackage{multirow}
\usepackage{mathtools}
\usepackage{float}
\usepackage[caption = false]{subfig}
\usepackage{xcolor,colortbl}
\usepackage{color}
\usepackage{tikz-cd}
\usetikzlibrary{arrows.meta}
\newcommand{\Tr}{\operatorname{Tr}}

\newcommand{\be}{\begin{equation}}
\newcommand{\ee}{\end{equation}}
\newcommand{\ba}{\begin{eqnarray}}
\newcommand{\ea}{\end{eqnarray}}
\newcommand{\ketbra}[2]{|#1\rangle \langle #2|}
\newcommand{\tr}{\operatorname{Tr}}
\newcommand{\proj}[1]{\ket{#1}\bra{#1}}

\newtheorem{theorem}{Theorem}
\newtheorem{corollary}{Corollary}

\newtheorem{proposition}{Proposition}
\newtheorem{observation}{Observation}

\newtheorem{remark}{Remark}

\def\>{\rangle}
\def\<{\langle}

\newtheorem{manualtheoreminner}{Theorem}
\newenvironment{manualtheorem}[1]{%
  \IfBlankTF{#1}
    {\renewcommand{\themanualtheoreminner}{\unskip}}
    {\renewcommand\themanualtheoreminner{#1}}%
  \manualtheoreminner
}{\endmanualtheoreminner}

\newcommand{\innerthmname}{}

\theoremstyle{definition}
\newtheorem*{innerdfn}{\innerthmname}
\newenvironment{upstatement}[1]
 {\renewcommand{\innerthmname}{#1}\innerdfn}
 {\endinnerdfn}

\newtheorem{manuallemmainner}{Lemma}

\newenvironment{manuallemma}[1]{%
  \IfBlankTF{#1}
    {}
    {\renewcommand\themanuallemmainner{#1}}%
  \manuallemmainner
}{\endmanuallemmainner}

\usepackage{centernot}
\usepackage{subfig}
\usepackage{filecontents}

\providecommand{\ket}[1]{| #1{\rangle}}
\providecommand{\bra}[1]{\langle #1|}

\providecommand{\proj}[1]{\ket{#1}\bra{#1}}
\providecommand{\ra}{\mathbf{r}_1}
\providecommand{\rb}{\mathbf{r}_2}
\providecommand{\bsa}{\boldsymbol{\sigma}_1}
\providecommand{\bsb}{\boldsymbol{\sigma}_2}
\providecommand{\bT}{\boldsymbol{T}}
\providecommand{\bR}{\boldsymbol{R}}
\providecommand{\bn}{\boldsymbol{n}}

\providecommand{\id}{\mathbf{I}}
\providecommand{\transp}[1]{#1^\intercal}
\providecommand{\trace}[1]{{\sf Tr}\left[#1\right]}
\providecommand{\bTT}{\boldsymbol{\Tilde{T}}}

\begin{document}

\title{Ergodiscord: An Operational and Distinct Notion of Quantumness of Correlations}

\author{Mir Alimuddin}
\affiliation{ICFO-Institut de Ciencies Fotoniques, The Barcelona Institute of Science and Technology, Av. Carl Friedrich Gauss 3, 08860 Castelldefels (Barcelona), Spain.}

\author{Snehasish Roy Chowdhury}
\affiliation{Physics and Applied Mathematics Unit, Indian Statistical Institute, 203 B.T. Road, Kolkata 700108}

\author{Ram Krishna Patra}
\affiliation{Department of Physics of Complex Systems, S. N. Bose National Center for Basic Sciences, Block JD, Sector III, Salt Lake, Kolkata 700106, India.}

\author{Subhendu B. Ghosh}
\affiliation{Department of Physics of Complex Systems, S. N. Bose National Center for Basic Sciences, Block JD, Sector III, Salt Lake, Kolkata 700106, India.}

\author{Tommaso Tufarelli}
\affiliation{School of Mathematical Sciences, University of Nottingham, University Park, Nottingham, NG7 2RD, United Kingdom.}

\author{Gerardo Adesso}
\affiliation{School of Mathematical Sciences, University of Nottingham, University Park, Nottingham, NG7 2RD, United Kingdom.}

\author{Manik Banik}
\affiliation{Department of Physics of Complex Systems, S. N. Bose National Center for Basic Sciences, Block JD, Sector III, Salt Lake, Kolkata 700106, India.}

\begin{abstract}
Nonclassicality in composite quantum systems depicts several puzzling manifestations, with Einstein-Podolsky-Rosen entanglement, Schr\"odinger steering, and Bell nonlocality being the most celebrated ones. In addition to those, an unentangled quantum state can also exhibit nonclassicality, as evidenced from notions such as quantum discord and work deficit. In this work, we present a general framework for exploring quantumness of correlations in multipartite quantum states. By exploiting the different signatures reflected on observable quantities depending on whether subsystems of a composite systems are probed jointly or independently, we introduce an operational quantifier of nonclassicality, termed {\em ergodiscord}. As we show, this newly proposed quantifier faithfully captures nonclassicality in any bipartite quantum state, while being fundamentally distinct from the original quantum discord. 
Moreover, ergodiscord uncovers an intriguing phenomenon called `nonlocal energy locking', where a useful form of energy (i.e. work) gets locked in correlations of nonclassical states. We also show that a mixed nonclassical state can lock more work than the maximally entangled state of the corresponding system, establishing an interesting super-additivity phenomenon of nonlocal energy locking. The present work may inspire novel designs of quantum energy storage devices by utilizing nonclassical correlations in composite quantum systems.
\end{abstract}


\maketitle	
\section{Introduction}
 Correlations are central to scientific inquiry across many disciplines, with their systematic study playing a key role in advancing our understanding of the physical world \cite{Pearl2009}. This is especially evident in the study of quantum systems comprising multiple subsystems. The intriguing nature of quantum correlations, as exemplified by entanglement and steering, has been recognized since the early foundations of quantum theory \cite{Einstein1935,Bohr1935,Schrdinger1935,*Schrdinger1936}. Building on these ideas, J. S. Bell, in his landmark work, extended the concept of quantum entanglement to establish quantum nonlocality \cite{Bell1964,*Bell1966} (see also \cite{Mermin1993}). Alongside their role in defining fundamental notions of nonclassicality \cite{Werner1989,Barrett2002,Wiseman2007}, quantum entanglement, steering, and nonlocality have found numerous applications in the rapidly advancing field of quantum technologies \cite{Horodecki2009,Guhne2009,Brunner2014,Uola2020}. More recently, it has been shown that quantum discord captures nonclassicality even in unentangled states \cite{Zurek2000,Ollivier2001,Henderson2001} (see also \cite{Modi2012,Adesso2016}). Consequently, the investigation of nonclassical correlations, with a focus on their identification, quantification, and classification, has become a central objective in quantum information research over the past few decades.  

In this paper, we adopt an operational approach to investigate nonclassical correlations in multipartite quantum systems. We start by noting that such systems can be probed differently depending on the resources available among the spatially separated subsystems. Global probing requires quantum communication channels among the different sub-parts, whereas in their absence, only local probing is feasible, where each subsystem is addressed individually. Allowing only classical communication enables an intermediate probing regime between the global and local cases (see Fig.~\ref{fig1}). The distinct observable outcomes resulting from these different probing strategies provide a foundation for defining and quantifying nonclassical correlations. This operational framework offers a versatile methodology to introduce various quantifiers of nonclassicality. While this approach can be built on resource theories leading to entropic measures, our framework also enables the construction of alternative quantifiers using expectation values of observable quantities. In particular, we propose a novel thermodynamic-based quantity, called `ergodiscord', which involves the works extractable from a multipartite system under reversible processes, that in general can be further constrained due to the availability of resources to implement them. As we show, for bipartite systems ergodiscord faithfully characterizes nonclassical correlations, vanishing only for states entirely devoid of such correlations. Our analysis also uncovers an operational phenomenon termed `nonlocal energy locking', wherein extractable work from an isolated multipartite quantum system becomes locked within its nonclassical correlations. We derive exact expressions for ergodiscord in generic two-qubit states, demonstrating its distinctness from quantum discord \cite{Zurek2000,Ollivier2001,Henderson2001} and suggesting that a resource theory of ergodiscord would differ qualitatively from that of quantum discord \cite{Modi2012,Adesso2016}. Finally, we discuss how this study can contribute to the design of efficient energy storage devices by exploiting the quantumness of correlations in multipartite states, with potential applications in quantum batteries \cite{Ferraro2018,Andolina2019,Yang2023}.

\section{Preliminaries}
 The state of an individual quantum system ($A$) is described by a density operator $\rho\in\mathcal{D}(\mathcal{H}_A)$. For a finite-dimensional system, which we will assume throughout this work, $\mathcal{H}_A$ is isomorphic to some complex Euclidean space $\mathbb{C}^{d_A}$, where $d_A$ denotes the dimension of the Hilbert space. The state of a composite system ${\bf A}\equiv A_1\cdots A_n$ is called separable if it is a statistical mixture of independently prepared constituent sub-parts, i.e., $\rho_{\bf A}=\sum_ip_i(\otimes_{j=1}^n\psi^i_j)$, where $\vec{p}:=\{p_i\}$ is a probability vector, and $\psi:=\ket{\psi}\bra{\psi}$ with $\ket{\psi^i}_j\in\mathbb{C}^{d_j}$. States that are not separable are called entangled, and they capture the most celebrated aspect of nonclassical correlation in composite quantum systems \cite{Horodecki2009,Guhne2009}. Importantly, a separable state can also accommodate nonclassical aspect recognized as quantum discord \cite{Zurek2000,Ollivier2001,Henderson2001}. A multipartite state is called classically correlated if it allows a decomposition of the form $\rho_{\bf A}=\sum p_{i_1\cdots i_n}\psi^{i_1}_1\otimes\cdots\otimes\psi^{i_n}_n$, where $\{\ket{\psi^{i_j}}_j\}$ forms an orthonormal basis (ONB) for the $j^{th}$ Hilbert space, and $\{p_{i_1\cdots i_n}\}$ denotes a probability vector. A function $\mathbb{F}:\mathcal{D}(\otimes_{j=1}^n\mathbb{C}^{d_j}_{A_j})\mapsto\mathbb{R}$ is said to capture the signature of nonclassical correlation faithfully if it takes value zero for all classically correlated states and nonzero otherwise. For the bipartite case, the entropic measure of quantum discord, defined in terms of the difference between two different expressions for the mutual information \cite{Zurek2000,Ollivier2001,Henderson2001}, is an example of a faithful quantifier of nonclassical correlation \cite{Adesso2016}. Later, this notion has been generalized to quantum systems involving more than two sub-parts, albeit with caveats in its operational interpretation \cite{Rulli2011,Xu2012,Radhakrishnan2020}. Notably, this notion of nonclassicality is party asymmetric, as a state can be a statistical mixture of product states, where $\{\ket{\psi^{i_j}}\}$ form ONBs only for some sub-parts, but not for all. In the following, we develop a general operational framework for studying nonclassical correlations in multipartite quantum states.

\section{Framework}
Any physical operation acting on a quantum system is described by a completely positive trace-preserving (CPTP) map \cite{Kraus1983}. Given a quantum system $A$, we start by defining a generic function $f_A:\mathcal{D}(\mathbb{C}^{d_A})\mapsto\mathbb{R}$. For a CPTP map $\phi_A:\mathcal{D}(\mathbb{C}^{d_A})\mapsto\mathcal{D}(\mathbb{C}^{d_A})$, we define the variation of the function $f_A$ on some state $\rho_A$ as $\Delta f_A^{\phi_A}(\rho_A):=f_A(\rho_A)-f_A(\phi(\rho_A))$. For a set of CPTP maps $\Phi_A$, the variation is given by $\Delta f_A^{\Phi_A}(\rho_A):=\max_{\phi_A\in\Phi_A}\Delta f_A^{\phi_A}(\rho_A)$, which quantifies the maximum extracted value of the function $f_A$ on the given state $\rho_A$ under the set of operations $\Phi_A$. For a multipartite quantum system ${\bf A}(=A_1\cdots A_n)$, fixing $(f_j,\Phi_j)$ pair for the $j^{th}$ party $j\in\{1,\cdots,n\}\equiv[n]$, and fixing $(f_{\bf A},\Phi_{\bf A})$ pair for the joint system, we define another function $\mathbb{F}_{GL}(\rho_{\bf A}):=\Delta f^{\Phi_{\bf A}}_{\bf A}(\rho_{\bf A})-\sum_j\Delta f^{\Phi_j}_{j}(\rho_j)$, where $\rho_j:=\tr_{A_j^c}(\rho_{\bf A})$ denotes the marginal state of the $j^{th}$ sub-part. The function $\mathbb{F}_{GL}$ quantifies the difference between the extracted value of the global function and that of the local functions evaluated on the global state and the local sub-parts under the fixed sets of operations. Interestingly, as we will see, for suitable choices of $(f_\gamma,\Phi_\gamma)$'s, with $\gamma\in\{[n]\cup{\bf A}\}$, the function $\mathbb{F}_{GL}$ can capture the signature of different forms of correlation within the multipartite state. As an example, for bipartite systems, the following observation is noteworthy (see Appendix I for more details).
\begin{figure}[t!]
\centering
\includegraphics[scale=0.45]{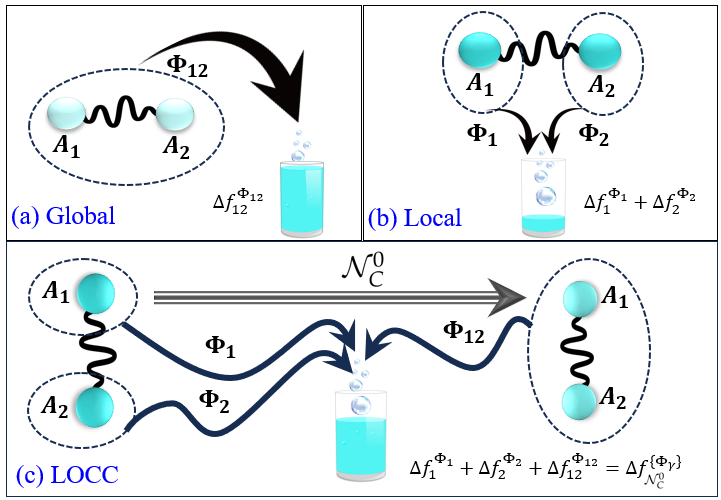}
\caption{Resource extraction from a bipartite quantum state $\rho_{\bf A}\in\mathcal{D}(\mathbb{C}^{d_1}\otimes\mathbb{C}^{d_2})$; here ${\bf A}=A_1A_2$. (a) $\Delta f^{\Phi_{12}}_{12}(\rho_{\bf A})$ quantifies the amount of resource extracted under the set of global operations $\Phi_{12}$. (b) $\sum_j\Delta f^{\Phi_j}_j(\rho_{j})$ is the amount of resource extracted under local operations $(\Phi_1,\Phi_2)$. (c) $\Delta f^{\{\Phi_{\gamma}\}}_{\mathcal{N}^0_C}(\rho_{\bf A})$ is the amount of resource extracted under $(\Phi_1,\Phi_2,\Phi_{12})$ when from the sub-part $A_1$ to $A_2$ free classical channel $\mathcal{N}^0_C$ are available. The quantity $\mathbb{F}_{G\mathcal{N}^0_C}(\rho_{\bf A}):=\Delta f^{\Phi_{12}}_{12}(\rho_{\bf A})-\Delta f^{\{\Phi_{\gamma}\}}_{\mathcal{N}^0_C}(\rho_{\bf A})$ amounts to 'nonlocal locking of the resource', and can capture the signature of nonclassical correlation in state $\rho_{\bf A}$.}
\label{fig1}\vspace{-.3cm}
\end{figure}
\begin{observation}\label{obs1}
The choices $(f_\gamma,\Phi_\gamma)\equiv(f_\gamma^{pe},\Phi_\gamma^{no}),~\forall~\gamma\in\{[2]\cup{\bf A}\}$ lead to $\mathbb{F}^{pe}_{GL}(\rho_{\bf A})=S(\rho_1)+S(\rho_2)-S(\rho_{\bf A}):=I(A_1:A_2)$; where $S(\rho_\gamma):=-\mathrm{tr}(\rho_\gamma\log\rho_\gamma)$, $f^{pe}_\gamma(\rho_\gamma):=\log d_\gamma-S(\rho_\gamma)$, and $\Phi^{no}_\gamma$ is the set of noisy operations.
\end{observation}
Notably, the function $f^{pe}_\gamma(\rho_\gamma)$ denotes the asymptotic rate of purity extraction from the state $\rho_\gamma$ under noisy operations \cite{Horodecki2003(1)}. Accordingly, $\mathbb{F}^{pe}_{GL}$ in Observation \ref{obs1}, quantifies the difference between asymptotic rates of purity extraction under global and local noisy operations \cite{Horodecki2003(2)}, which interestingly finds an exact match with the total correlation of the state $\rho_{\bf A}$, quantified through the quantum mutual information $I(A_1:A_2)$ \cite{Groisman2005}.

At this point it is worth noting how the function $\mathbb{F}_{GL}$ be evaluated operationally for a multipartite state (see Fig.~\ref{fig1}). Firstly, each of the local parties can evaluate the corresponding $\Delta f^{\Phi_j}_j$ locally. On the other hand, for evaluating $\Delta f^{\Phi_{\bf A}}_{\bf A}$, the spatially separated parties can send their respective sub-parts to a common party (say $j^\star$). This however demands a perfect quantum communication channel from all the other parties to the designated party $j^\star$. At this point, one can consider intermediate scenarios where only a subset of channels $\mathcal{N}$ are available from the remaining parties to $j^\star$. In such cases one can define the variation $\Delta f^{\{\Phi_\gamma\}}_{\mathcal{N}}(\rho_{\bf A}):=\max_{\{\phi_\gamma\in\Phi_\gamma,~N\in\mathcal{N}\}}[\sum_{j\in[n]}\Delta f^{\phi_j}_j(\rho_j)+\Delta f^{\phi_{\bf A}}_{\bf A}(N(\sigma_{\bf A}))-\mathcal{C}(N)]$, where $\sigma_{\bf A}:=(\otimes_j\phi_j)(\rho_{\bf A})$, $N(\sigma_{\bf A})$ is the composite state received by $j^\star$, and $\mathcal{C}(N)$ is a suitably defined cost function. This function quantifies the value extraction from a given $\rho_{\bf A}$ over the maps $\{\Phi_{\bf A}\}$'s when collaboration from the remaining parties to $j^\star$ is limited by the set of maps $\mathcal{N}$. To say more elaborately, each of the parties extracts $\Delta f^{\phi_j}_j$ from their respective parts under their local operations resulting the state $\rho_{\bf A}$ to evolve into a new state $\sigma_{\bf A}$. Accordingly, the $j^\star$ party receives the state $N(\sigma_{\bf A})$ and extracts $\Delta f^{\phi_{\bf A}}_{\bf A}$. Clearly, in absence of $\mathcal{N}$, the value of $\Delta f^{\{\Phi_\gamma\}}_{\mathcal{N}}$ should exactly match with $(\sum_{j\in[n]}\Delta f^{\Phi_j}_j)$, whereas with perfect quantum communication channel it should match with $\Delta f^{\Phi_{\bf A}}_{\bf A}$. The cost function $\mathcal{C}(N)$ is defined accordingly to satisfy these demands \cite{Bedingham2016,Faist2018}. At this point, we  define a new function $\mathbb{F}_{G\mathcal{N}}(\rho_{\bf A}):=\Delta f^{\Phi_{\bf A}}_{\bf A}(\rho_{\bf A})-\Delta f^{\{\Phi_\gamma\}}_{\mathcal{N}}(\rho_{\bf A})$ quantifying the difference between optimal global extracted value and the $\mathcal{N}$ channel assisted optimal extracted value, which will be shown to be useful to explore nonclassical correlation in multipartite states.

\section{Resource theoretic quantifiers}
To see application of the aforesaid general framework we start by analyzing how the function $\mathbb{F}_{G\mathcal{N}}$ can capture the notion of quantum discord in bipartite systems (${\bf A}=A_1A_2$) for suitable choices of the tuple $(f_\gamma,\Phi_\gamma,\mathcal{N})$. For that once again we fix $(f_\gamma,\Phi_\gamma)\equiv(f^{pe}_\gamma,\Phi_\gamma^{no}),~\forall~\gamma\in\{1,2,{\bf A}\}$, and fix $\mathcal{N}$ to be classical communication channel between $A_1$ and $A_2$. Recall that quantum discord is a party asymmetric concept: a bipartite state $\rho_{\bf A}\in\mathcal{D}(\mathbb{C}^{d_1}\otimes\mathbb{C}^{d_2})$ having the classical-quantum (CQ) form $\rho_{\bf A}=\sum p_{i_1i_2}\psi^{i_1}_1\otimes\psi^{i_2}_2$, with $\langle\psi^{i_1}_1|\psi^{i^\prime_1}_1\rangle=\delta_{i_ii^\prime_1}$, has discord zero from $A_1\to A_2$  \cite{Zurek2000,Ollivier2001,Henderson2001}. A generic classical communication channel $N\in\mathcal{N}_C$ sending the $A_1$ part to the $A_2$ can be thought as a measure-and-prepare channel \cite{Holevo1998,Rosset2018}, with preparations being restricted to rank-1 orthogonal states (see Appendix III). Such a channel could map completely mixed state to a resourceful state, and therefore is costly from the perspective of purity resource theory. One can bypass this cost association by further limiting these channels to be cost free. As we argue in Appendix III, a cost-free classical channel in this case takes the form $N(\rho_{A_1}):=\sum_{l\in[d_1]}\tr(\pi_l\rho_{1})\ket{\alpha^l_1}\bra{\alpha^l_1}$, where $\langle\alpha^l_1|\alpha^{l^\prime}_1\rangle=\delta_{ll^\prime}$, and $\{\pi_l\}$ is a positive-operator-value-measurement (POVM) with $\tr(\pi_l)=1,~\forall~l\in[d_1]$. By denoting the set of all such zero cost classical channels as $\mathcal{N}^0_C$, we obtain the following result.
\begin{theorem}\label{theo1}
For $(f_\gamma,\Phi_\gamma,\mathcal{N})\equiv(f_\gamma^{pe},\Phi_\gamma^{no},\mathcal{N}^0_C)$ the function $\mathbb{F}^{pe}_{G\mathcal{N}^0_C}$ faithfully captures nonclassical correlation from $A_1\to A_2$ in any bipartite state $\rho_{\bf A}\in\mathcal{D}(\mathbb{C}^{d_1}\otimes\mathbb{C}^{d_2})$ having maximally mixed marginal on $A_1$ sub-part.
\end{theorem}
Proof is provided in Appendix III. During the process we also prove the following Proposition which is interesting on its own right (see Appendix II).
\begin{proposition}\label{prop1}
For $\Lambda:\mathcal{D}(\mathbb{C}^d)\to\mathcal{D}(\mathbb{C}^d)$ being a CPTP unital map, the states $\rho$ and $\Lambda(\rho)$ have the same spectrum whenever $S(\Lambda(\rho))=S(\rho)$.
\end{proposition}
At this point we would like to point out that a special case of our Theorem \ref{theo1} reproduces one of Zurek's results \cite{Zurek2003}. While the set of free channels there consists of measure-and-orthogonal preparation with the measurement being rank-1 projective measurements, in our case we consider more general kind of measurements. Notably, similar analysis as of Theorem \ref{theo1} is worth exploring by fixing $(f_\gamma,\Phi_\gamma,\mathcal{N})\equiv(f_\gamma^{fe},\Phi_\gamma^{to},\mathcal{N}_{tc})$, where $f_\gamma^{fe}$and $\Phi_\gamma^{to}$ respectively denote thermodynamic free energy and thermal operations for a given inverse temperature $\beta$, and $\mathcal{N}_{tc}$ are thermal classical channels compatible to this thermodynamic resource theory. The resulting study will go in line with the works in Refs.~\cite{Oppenheim2002,Brandao2013,Perarnau2015,Morris2019}.

\section{Generic nonclassicality quantifiers and Ergodiscord}
We now proceed to explore the generality of our framework to construct quantifiers for nonclassical correlations that need not to be motivated from a resource theory. To this aim, we start by noting that the function $f:\mathcal{D}(\mathbb{C}^{d})\mapsto\mathbb{R}$ can simply be the expectation value of some observable $\mathcal{O}\in\mbox{Herm}(\mathbb{C}^d)$, {\it i.e.}, $\mathcal{O}:\rho\mapsto\langle\mathcal{O}\rangle_{\rho}:=\tr(\mathcal{O}\rho)$. Unlike the earlier case, here the variation is evaluated over the set of reversible operations, resulting $\Phi$ to be the set of unitary operations $\mathrm{U}$. In this case the variation $\Delta\mathcal{O}^{\mathrm{U}}(\rho):=\max_{u\in\mathrm{U}}[\tr(\mathcal{O}\rho)-\tr(\mathcal{O} u\rho u^\dagger)]$ quantifies the optimal  value of the observable $\mathcal{O}$ extracted from the state $\rho$ when the system is evolved under closed system dynamics. While studying multipartite systems, generality of our framework allows choosing different observables for different sub-parts. Although the choice of the global observable can also be arbitrary, here we consider $\mathcal{O}_{\bf A}\equiv\sum_{j\in[n]}\tilde{\mathcal{O}}_j~$, where $\tilde{\mathcal{O}}_j:=\mathcal{O}_j\otimes\mathbf{I}_{rest}$. From here onward, we, without loss of any generality, will consider the observable to be the Hamiltonian of the system. Accordingly, the variation $\Delta h^{\mathrm{U}}(\rho)$ becomes the ergotropy of the state, which has a long history in literature \cite{Pusz1978,Lenard1978,Allahverdyan2004}. The corresponding function $\mathbb{H}_{GL}(\rho_{\bf A})$ quantifies the ergotropic gap -- the difference between the extractable work from a state $\rho_{\bf A}$ under global and local unitary operations, respectively. Interestingly, this function and its variants can capture different forms of nonclassical correlations, particularly entanglement in bipartite states and genuine entanglement in multipartite quantum states \cite{Mukherjee2016,Alimuddin2019,Alimuddin2020,Puliyil2023,Joshi2024}.
\begin{figure}[t!]
\centering
\includegraphics[width= 0.488\textwidth]{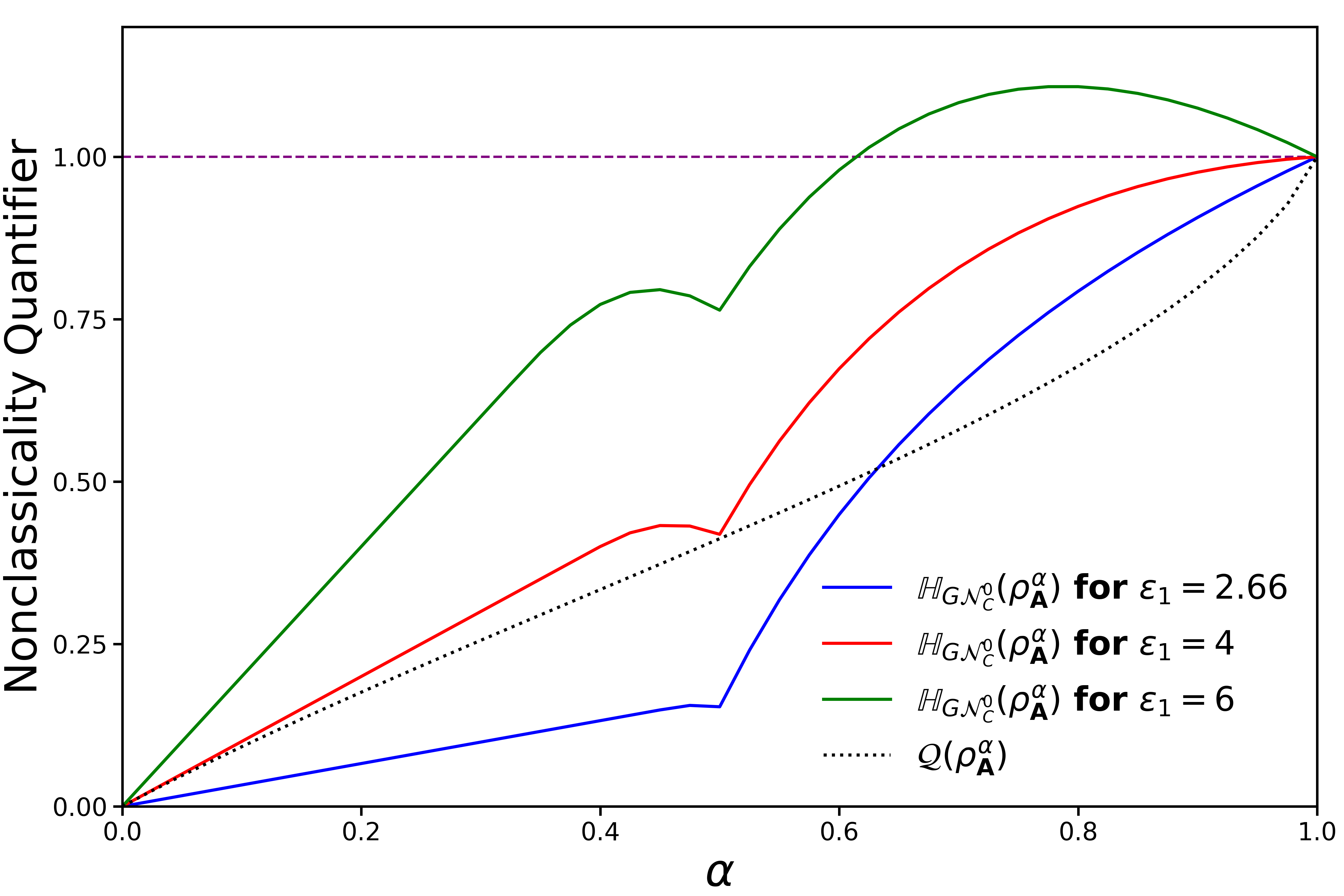}
\caption{Nonlocal energy locking and superadditivity phenomena. Consider the one parameter family of states \(\rho^\alpha_{\bf A}\). Black dotted curve represents quantum discord \(\mathcal{Q}(\rho^\alpha_{\bf A})\) \cite{Zhu2018}. Fixing \(\epsilon_2=2\), the ergodiscord \(\mathbb{H}_{G\mathcal{N}^0_C}(\rho^\alpha_{\bf A})\) is plotted for different values of $\epsilon_1$. For \(\epsilon_1=6\), the value of ergodiscord exceeds unity for certain range of the parameter \(\alpha\). Notably, the corresponding states lock more energy than the maximally entangled state $\ket{\phi^{\max}}$. These states when compared with maximally entangled state, depicts the feature as stated in Theorem \ref{theo4}.}
\label{fig2}\vspace{-.3cm}
\end{figure}

While studying the quantity $\mathbb{H}_{G\mathcal{N}^0_C}$, which is one of our core interests in the present work, we first need to identify the free classical channels $\mathcal{N}^0_C$. As it can be argued a classical channel $N\in\mathcal{N}_C$ communicating $A_j$ sub-part to the $j^\star$ party will be considered to be free if $\tr(h_jN(\rho_j))\le\tr(h_j\rho_j)~\&~\Delta h^{U_j}_jN(\rho_j)\le\Delta h^{U_j}_j(\rho_j),~\forall~\rho_j\in\mathcal{D}(\mathbb{C}^{d_j})$. A more elaborate justification of this demand as well as a complete characterization of free classical channels $\mathcal{N}^0_C$ are provided in  Appendix IV. With this we now present a general result for bipartite systems (proof is deferred to Appendix IV).
\begin{theorem}\label{theo2}
The function $\mathbb{H}_{G\mathcal{N}^0_C}(\rho_{\bf A})$, termed {\bfseries ergodiscord}, faithfully captures nonclassical correlation from $A_1\to A_2$ in any bipartite state $\rho_{\bf A}$, whenever $h_{\bf A}$ has non degenerate spectrum.
\end{theorem}
Notably unlike Theorem \ref{theo1}, the function $\mathbb{H}_{G\mathcal{N}^0_C}(\rho_{\bf A})$ faithfully captures signature of nonclassical correlation in any bipartite state without assuming any further constraint on the state. Theorem \ref{theo2} has quite interesting implications. A nonzero value of $\mathbb{H}_{G\mathcal{N}^0_C}(\rho_{\bf A})$ implies that the extractable work under global operations is strictly greater than the work extractable under local operations assisted with classical communication from $A_1$ to $A_2$. In other words, work gets locked in the nonclassical correlation of the state. This phenomenon we call the  {\it nonlocal energy locking} in nonclassical correlation. The motivation of this name arises from the pioneering `nonlocality without entanglement' phenomenon, where distinguishability of a set of multipartite product states under local operation and classical communication is suboptimal than the global distinguishability \cite{Peres1991,Bennett1999} (see also \cite{Halder2019,Bhattacharya2020,Banik2021,Rout2021}), and the works \cite{DiVincenzo2004,Horodecki2005,Smolin2006} that report several locking phenomena utilizing quantum resources). The following crucial property of the quantity discussed in Theorem \ref{theo2} is worth noting.
\begin{proposition}\label{prop2}
$\mathbb{H}_{G\mathcal{N}^0_C}(\rho_{\bf A})=\mathbb{H}_{G\mathcal{N}^0_C}(\sigma_{\bf A})$, where $\sigma_{\bf A}=u_{1}\otimes u_2(\rho_{\bf A})u^\dagger_{1}\otimes u^\dagger_2$ for arbitrary $u_i$'s on the respective parts.
\end{proposition}
\begin{proof}
While the Proposition holds for any observable $\mathbb{O}$, here we provide the proof by considering it to be the system's Hamiltonian. Here we will use a property of ergodiscord which we prove for general observables in the Appendix (see equation (\ref{del})).
\begin{align}
\mathbb{H}_{GN^0_c}(\rho_{{\bf A}}) = \min_{U_1} E(\widetilde{\chi}_{\bf A}^p) - E(\rho_{\bf A}^p),   
\end{align}
where $\rho^p_{\bf A}$ and $\widetilde{\chi}^p_{\bf A}$ respectively denotes the passive state corresponding to $\rho_{\bf A}$ and $N_c^o \circ U_1 (\rho_{\bf A})$. Let $U_1^\star$ be the optimizing unitary operation. Then, the expression simplifies to 
\begin{align}\label{pr21}
\mathbb{H}_{GN^0_c}(\rho_{\bf A}) = E(\chi_{\bf A}^p) - E(\rho_{\bf A}^p),  
\end{align}
with $\chi^p_{\bf A}$ denoting the passive state corresponding to $N_c^o \circ U^\star_1 (\rho_{\bf A})$. The evolution of the state is thus given by

\begin{equation}\label{pr22}
\begin{tikzcd}[row sep=large, column sep=small, scale=1, transform shape]
  \rho_{\mathbf A}
    \arrow[r, "U^\star_1"]
  & \tau_{\mathbf A}
    = U^\star_1\,\rho_{\mathbf A}\,(U^\star_1)^\dagger
    \arrow[r, "N^0_c"]
  & \chi_{\mathbf A}
    = N^0_c(\tau_{\mathbf A})
    \arrow[d, "U_{\mathbf A}"']
  \\
  & & \chi^p_{\mathbf A}
    = U_{\mathbf A}\,\chi_{\mathbf A}\,U_{\mathbf A}^\dagger
\end{tikzcd}
\end{equation}

Let $\sigma_{\bf A} = (V_1 \otimes W_2) \rho_{\bf A} (V_1^\dagger \otimes W_2^\dagger)$, where $V_1$ and $W_2$ are some arbitrary local unitary on the $1$ and $2$ sub-parts respectively. The expression (\ref{pr21}) in this case reads as,
\begin{align}
\mathbb{H}_{GN^0_c}(\sigma_{\bf A}) = E(\eta_{\bf A}^p) - E(\sigma_{\bf A}^p),   
\end{align}
and the state evolution [as describe in Eq.(\ref{pr22})] in this case becomes,

\begin{equation}\label{pr23}
\begin{tikzcd}[row sep=large, column sep=small, scale=1, transform shape]
  \sigma_{\mathbf A}
    \arrow[r, "Z^\star_1"]
  & \zeta_{\mathbf A}
    = Z^\star_1\,\sigma_{\mathbf A}\,(Z^\star_1)^\dagger
    \arrow[r, "N^0_c"]
  & \eta_{\mathbf A}
    = N^0_c(\zeta_{\mathbf A})
    \arrow[d, "U'_{\mathbf A}"']
  \\
  & & \eta^p_{\mathbf A}
    = U'_{\mathbf A}\,\eta_{\mathbf A}\,(U'_{\mathbf A})^\dagger
\end{tikzcd}
\end{equation}

Note that $\sigma^{p}_{\bf A} = \rho^p_{\bf A}$ as $\rho_{\bf A}$ and $\sigma_{\bf A}$ are unitarily connected. Therefore, to establish the claim of the proposition, we only need to examine the invariance of $E(\eta^p_{\bf A})$ and $E(\chi^p_{\bf A})$. The proof proceeds by discussing these aspects in two parts.
\begin{itemize}
\item[{\bf P-I:}] Consider that $\sigma_{\bf A} = (I_1 \otimes W_2) \rho_{\bf A} (I_1 \otimes W_2^\dagger)$. The optimal state $\eta_{\bf A}$ achieved in Eq.(\ref{pr23}) when starting with the initial state $\sigma_{\bf A}$ can also be achieved staring with the initial state $\rho_{\bf A}$. For that the second party applies the unitary $W_2$ on the state $\chi_{\bf A}$ in Eq.(\ref{pr22}), i.e., $\eta_{\bf A} = (I_1 \otimes W_2) \chi_{\bf A} (I_1 \otimes W_2^\dagger)$, implying $E(\chi^p_{\bf A}) = E(\eta^p_{\bf A})$.
\item[{\bf P-I:}] Consider that $\sigma_{\bf A}= (V_1 \otimes I_2) \rho_{\bf A} (V_1^\dagger \otimes I_2^\dagger)$. Assume that, $E(\chi^p_{\bf A})> E(\eta^p_{\bf A})$. On the other hand, we have 
\begin{align}
\rho_{\bf A}\xrightarrow{V_1} \sigma_{\bf A}\xrightarrow{Eq.(\ref{pr23})}\eta^p_{\bf A}.    
\end{align}
This implies that starting from $\rho_{\bf A}$, they can reach to the passive state $\eta^p_{\bf A}$. This along with the assumption $E(\chi^p_{\bf A})> E(\eta^p_{\bf A})$ contradicts the claim that $U^\star_1$ in Eq.(\ref{pr22}) is the optimal unitary, and thus $E(\chi^p_{\bf A})\ngtr E(\eta^p_{\bf A})$.

Now, assume that $E(\chi^p_{\bf A})< E(\eta^p_{\bf A})$. Again this will contradict the optimality of $Z^\star_1$ in Eq.(\ref{pr23}). Therefore we have, $E(\chi^p_{\bf A})=E(\eta^p_{\bf A})$.
\end{itemize}
(P-I) and (P-II) together thus implies $\mathbb{H}_{GN^0_c}(\rho_{\bf A}) = \mathbb{H}_{GN^0_c}\big((V_1 \otimes W_2) \rho_{\bf A} (V_1^\dagger \otimes W_2^\dagger)\big)$. This completes the proof.
\end{proof}
\section{Ergodiscord for general two-qubit states}
A simplified from of $\mathbb{H}_{G\mathcal{N}^0_C}(\rho_{\bf A})$ can be obtained for generic two-qubit state $\rho_{\bf A}=\frac{1}{4}(\mathbf{I}+\boldsymbol{\sigma}^\intercal_1{\bf r}_1+{\bf r}^\intercal_2\boldsymbol{\sigma}_2+\boldsymbol{\sigma}^\intercal_1{\bf T}\boldsymbol{\sigma}_2)$.
\begin{theorem}\label{theo3}
$\mathbb{H}_{G\mathcal{N}^0_C}(\rho_{\bf A})=\min_{\bf m}\sum_{k=1}^4E_k\left(q_k({\bf m})-p_k\right)$; where $\{E_k\}$ are the non-degenerate increasing eigenvalues of $\mathbb{H}$, $\{p_k\}$ are the decreasing spectral values of $\rho_{\bf A}$, and $\{q_k({\bf m})\}$'s take values from the set $\{{\footnotesize \frac{1}{4}\left(1+(-1)^i{\bf m}^\intercal{\bf r}_1+(-1)^j\left|{\bf r}^\intercal_2+(-1)^i{\bf m}^\intercal{\bf T}\right|\right)}\}_{i,j=0}^1$ in decreasing order, and ${\bf m}$ is an arbitrary Bloch vector with $|{\bf m}|=1$.
\end{theorem}
\begin{proof}
A general 2-qubit state may be written as
\begin{align}
\rho=\frac{1}{4}\left[\id+\transp{\bsa}\ra+\transp{\rb}\bsb+\bsa^\intercal\bT\bsb\right],
\end{align}
where $\ra$ is the Bloch vector of subsystem $A,$ $\bsa=(\sigma_x,\sigma_y,\sigma_z)_A$ is the vector of 3 Pauli matrices for subsystem $A$ (same idea for $B$) and $\bT$ is the correlation matrix. These can be obtained via:
\begin{align}
    \ra=\trace{\rho\bsa},~~~\rb=\trace{\rho\bsb},~~~ \bT=\trace{\rho\bsa\transp{\bsb}}.
\end{align}
Application of local unitary $U_1$ is equivalent to applying a 3x3 real rotation matrix to $\bsa$, or equivalently replacing $\ra\to\bR\ra$ and $\bT\to\bR\bT$.
Then we have
\begin{align}
    \tau=U_1\rho U_1^\dagger=\frac{1}{4}\left[\id+\transp{\bsa}\bR\ra+\transp{\rb}\bsb+\bsa^\intercal\bR\bT\bsb\right],
\end{align}
The classical channel $\mathcal{N}^0_C$ between Alice and Bob is a measure and prepare channel. Any generic pair of qubit projectors can be written as
$ P_{\pm\bn}=\proj{\bn_\pm}=\frac{1}{2}\left(\id\pm\transp{\bsa}\bn\right),$
where $\bn$ is a unit vector in the Bloch sphere (it identifies the 2 orthogonal pure states forming the eigenbasis).
This projective measurement also follows $\bra{\bn_\pm}\transp\bsa\ket{\bn_\pm}=\pm\transp\bn$.
So the projected state of Bob after passing through $\mathcal{N}^0_C$ is
\begin{align}
\tau_\pm&=\bra{\bn_\pm}\tau\ket{\bn_\pm}_A\nonumber\\
&=\frac{1}{4}\left[\id\pm\transp{\bn}\bR\ra+\left(\transp\rb\pm\bn^\intercal\bR\bT\right)\bsb\right],
\end{align}
The state accessible to Bob after all the above operations is $\chi=P_{\bn}\tau_{\bn}+P_{-\bn}\tau_{-\bn}$. We know that the eigenvalues of a qubit Hermitian operator of the form $B=b_0\id+\transp{\boldsymbol{b}}\bsb$ are $\lambda_{\pm}=b_0\pm|\boldsymbol{b}|$.
So we get the eigenvalues of the partial projections:
\begin{align*}
{\sf Spectrum}[\tau_{+\bn}]&=\left\{\frac{1}{4}\left[1+\transp{\bn}\bR\ra\pm\left|\transp\rb+\bn^\intercal\bR\bT\right|\right]\right\},\\
{\sf Spectrum}[\tau_{-\bn}]&=\left\{\frac{1}{4}\left[1-\transp{\bn}\bR\ra\pm\left|\transp\rb-\bn^\intercal\bR\bT\right|\right]\right\}.
\end{align*}
Since the projectors $P_{\pm\bn}$ are mutually orthogonal, the eigenvalues of $\chi$ follow directly from our previous results:
\begin{align}
{\sf Spectrum}[\chi]=\left\{\!\begin{aligned}
\frac{1}{4}\left[1+\transp{\bn}\bR\ra+\left|\transp\rb+\bn^\intercal\bR\bT\right|\right],\\~\frac{1}{4}\left[1+\transp{\bn}\bR\ra-\left|\transp\rb+\bn^\intercal\bR\bT\right|\right]\\
\frac{1}{4}\left[1-\transp{\bn}\bR\ra+\left|\transp\rb-\bn^\intercal\bR\bT\right|\right],\\~\frac{1}{4}\left[1-\transp{\bn}\bR\ra-\left|\transp\rb-\bn^\intercal\bR\bT\right|\right]\end{aligned}\right\}.
\end{align}
We note that $\bR$ always appears in the form $\transp{\bn}\bR$, so we may simplify the problem by defining a new unit vector ${\bf m}\equiv \transp{\bR}\bn$. Now, we only need to optimize over the vector ${\bf m}$, under the constraint $|{\bf m}|=1$, so this gives us 
\begin{align}
 &{\sf Spectrum}[\chi]=\{s_{ij}({\bf m}):=\nonumber\\ 
 &\frac{1}{4}(1+(-1)^i{\bf m}^\intercal{\bf r}_1+(-1)^j|{\bf r}^\intercal_2+(-1)^i{\bf m}^\intercal{\bf T}|)\}_{i,j=0}^1.
\end{align}
After arranging the spectrum in decreasing order we get $\{s_{ij}({\bf m})\}\mapsto \{q_k({\bf m})$\}. Let $\{E_k\}$ be the non-degenerate increasing eigenvalues of $\mathbb{H}$, $\{p_k\}$ are the decreasing spectral of $\rho_{\bf A}$. Then the extractable ergotropy is.
\begin{align}
H_{GN^0_c}(\rho_{\bf A}) &= \min_{u_1} E({\chi}_{\bf A}^p) - E(\rho_{\bf A}^p)\nonumber\\&=\min_{m}\sum_{k=1}^4(q_k({\bf m})-p_k)E_k.
\end{align}
This completes the proof.
\end{proof}

We also obtain a closed letter expression of this quantity for the states having completely mixed marginals which is stated in the following Corollary.
\begin{corollary}\label{coro1}
For a two-qubit state having completely mixed marginals,  $\mathbb{H}_{G\mathcal{N}^0_C}(\rho_{\bf A})=\omega(E_1+E_2)/2+(1-\omega)(E_3+E_4)/2-\sum_kp_kE_k$, where $2\omega:=1+\sqrt{\lambda}$, with $\lambda$ being the maximum eigenvalue of the matrix ${\bf TT}^\intercal$.
\end{corollary}
\begin{proof}
A most general two-qubit state with maximally mixed marginals can be written as $\rho_{\bf A}=\frac{1}{4}\left[\id+\bsa^\intercal\bT\bsb\right]$. This state can always be written as convex mixture of Bell state by applying local unitary $V_1\otimes W_2$, i.e.,
\begin{align*}
\sigma_{\bf A}&=(V_1\otimes W_2)\rho_{\bf A}(V_1^{\dagger}\otimes W_2^{\dagger})=\frac{1}{4}\left[\id+\bsa^\intercal\bTT\bsb\right]\\
&=p_1\ket{\phi^+}\bra{\phi^+}+p_2\ket{\phi^-}\bra{\phi^-}+p_3\ket{\psi^+}\bra{\psi^+}\\&\hspace{2cm}+p_4\ket{\psi^-}\bra{\psi^-},
\end{align*}
where $\bTT ={\sf diag}(2(p_1+p_3)-1,2(p_2+p_3)-1),2(p_1+p_2)-1)$. Denoting, $\omega=\max\{(p_1+p_3),(p_2+p_3),(p_1+p_2)\}$, we have $\sqrt{\lambda_{\max}\left(\bT\transp\bT\right)}=2\omega-1$. Let us arrange the probabilities in decreasing order, $\{p_i\}\mapsto\{p^\downarrow_i\}$ with  $p^\downarrow_1\geq p^\downarrow_2\geq p^\downarrow_3 \geq p^\downarrow_4$. The energy of the $\sigma_{\bf A}^p$ and $\chi_{\bf A}^p$ are given by
\begin{align*}
E(\sigma_{\bf A}^p) =&\sum_{i=1}^4p_i^\downarrow E_i,\\
E(\chi_{\bf A}^p)=&\frac{1}{4}\left(1+\sqrt{\lambda_{\max}\left(\bT\transp\bT\right)}\right)(E_1+E_2) \\&+\frac14\left(1-\sqrt{\lambda_{\max}\left(\bT\transp\bT\right)}\right)(E_3+E_4)\\
=&\frac{\omega}{2}(E_1+E_2)+\frac{(1-\omega)}{2}(E_3+E_4).
\end{align*}
Thus we have,
\begin{align*}
\mathbb{H}_{GN^0_c}(\rho_{\bf A})=\frac{\omega}{2}(E_1+E_2)+\frac{1-\omega}{2}\left(E_3+E_4\right)-\sum_ip_i^\downarrow E_i.
\end{align*}
This completes the proof.
\end{proof}

\begin{figure}[tbh!]
\centering
\includegraphics[width=0.488\textwidth]{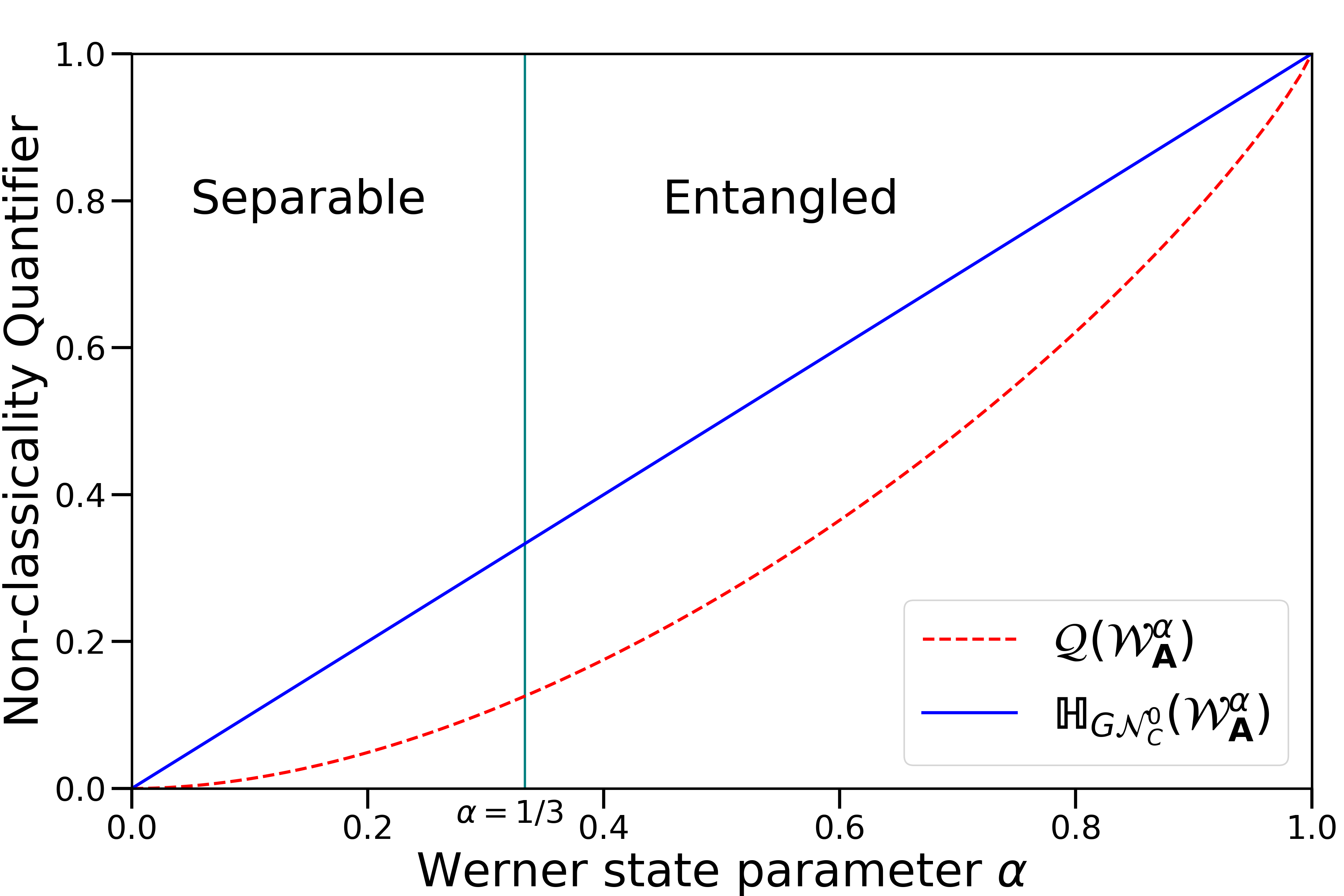}
\caption{Nonlocal energy locking in two qubit Werner states $\mathcal{W}^\alpha_{\bf A}$. Red dashed curve denotes the original quantum discord $\mathcal{Q}(\mathcal{W}^\alpha_{\bf A})=\frac{1-\alpha}{4}\log(1-\alpha)-\frac{1+\alpha}{2}\log(1+\alpha)+\frac{1+3\alpha}{4}\log(1+3\alpha)$ establishing the states are nonclassical for $\alpha\in(0,1]$ \cite{Ollivier2001} (see also \cite{Luo2008}). Blue curve amounts to the ergodiscord $\mathbb{H}_{G\mathcal{N}^0_C}(\mathcal{W}^\alpha_{\bf A})~[=\alpha\epsilon_2/2]$, i.e.~the energy locked into the nonclassical correlation of the state $\mathcal{W}^\alpha_{\bf A}$. Here the plot is given by fixing $\epsilon_2=2$.}
\label{figs2}\vspace{-.3cm}
\end{figure}

Using this result we derive and plot the ergodiscord for Werner class of state along with its quantum discord. Corollary 1, when applied to Werner class of states $\mathcal{W}^{\alpha}_{\bf A}= \alpha \phi^{max}_{\bf A}+(1-\alpha)\mathbf{I}_{\bf A}/4$, letting the Hamiltonian to be $h_{\bf A}=h_{A_1}\otimes\mathbf{I}_{A_2}+\mathbf{I}_{A_1}\otimes h_{A_2}$, with $h_{A_1}:=\epsilon_1\ket{1}\bra{1}~\&~h_{A_2}:=\epsilon_2\ket{1}\bra{1}$, we obtain $\mathbb{H}_{G\mathcal{N}^0_C}(\mathcal{W}^\alpha_{\bf A})=\alpha\epsilon_2/2$. Following  facts are noteworthy:

The ergodiscord of the Werner state is independent of $\epsilon_1$, indicating that the energy difference between the two subsystems does not influence nonlocal energy locking. This is in stark contrast to the one-parameter family of states $\rho_{\bf A}^{\alpha}$ (see Figure \ref{fig2}), where increasing the energy gap between subsystems significantly alters the nature of energy locking.

From Figure \ref{figs2}, it is evident that for the Werner class of states, both quantum discord and ergodiscord increase monotonically with the parameter $\alpha$. However, this behavior does not hold for $\rho^{\alpha}_{\bf A}$. Figure \ref{fig2} reveals that within certain ranges of $\alpha$, quantum discord exhibits monotonic growth, while ergodiscord decreases. Identifying the class of states for which the monotonicity of these two quantifiers aligns presents an intriguing avenue for further exploration.

\section{Independence of ergodiscord and quantum discord}
While both ergodiscord and quantum discord faithfully capture non-classicality of bipartite quantum states, our next result shows that they reflect independent notions.
\begin{theorem}\label{theo4}
There exist bipartite states $\rho_{\bf A},\sigma_{\bf A}\in\mathcal{D}(\mathbb{C}^{2}\otimes\mathbb{C}^{2})$ such that, $\mathcal{Q}(\rho_{\bf A})<\mathcal{Q}(\sigma_{\bf A})$ but $\mathbb{H}_{G\mathcal{N}^0_C}(\rho_{\bf A})>\mathbb{H}_{G\mathcal{N}^0_C}(\sigma_{\bf A})$.
\end{theorem}
\begin{proof}
Consider the one parameter family of states \(\rho^\alpha_{\bf A}:=\alpha\phi^{\max}_{\bf A}+(1-\alpha)\ket{01}_{\bf A}\bra{01}\), with \(\alpha\in[0,1]\), $\phi^{\max}_{\bf A}:=\ket{\phi^{\max}}_{\bf A}\bra{\phi^{\max}}$, and \(\ket{\phi^{max}}:=(\ket{00}+\ket{11})/\sqrt{2}\). Let system's Hamiltonian be $h_{\bf A}=h_{A_1}\otimes\mathbf{I}_{A_2}+\mathbf{I}_{A_1}\otimes h_{A_2}$, with $h_{A_1}:=\epsilon_1\ket{1}\bra{1}~\&~h_{A_2}:=\epsilon_2\ket{1}\bra{1}$. As depicted in Fig. \ref{fig2}, for suitable choices of $\epsilon_1,\epsilon_2$, the states possess the claimed feature. 
\end{proof}
Not only does ergodiscord captures a distinct notion of nonclassicality as compared to the original quantum discord, but the set of free operations appropriate for bipartite quantum discord is also inadequate for it. Consider a bipartite system \(\mathbf{A}\) composed of two subsystems \(A_1\) and \(A_2\). It is well established that, within the resource theory of quantum discord from \(A_1 \mapsto A_2\), the set of free operations must, at a minimum, include all completely positive trace-preserving (CPTP) maps acting on the \(A_2\) subsystem \cite{Piani2012}. However, this does not hold in the case of ergodiscord as established in the following.

\begin{proposition}\label{prop3}
For every $\alpha\in[0,1]$, there exist a bipartite quantum state $\ket{\psi^\alpha}=\sqrt{1-\alpha/2}\ket{00}+\sqrt{\alpha/2}\ket{11}$ and a CPTP map $\Lambda^\alpha$ of the form
\begin{align*}
&\Lambda^\alpha\left(\ketbra{0}{0}\right) = \frac{\alpha}{2-\alpha}\ketbra{0}{0}+\frac{2(1-\alpha)}{2-\alpha}\ketbra{1}{1},\\ &\Lambda^\alpha\left(\ketbra{0}{1}\right)= \frac{\alpha}{\sqrt{\alpha(2-\alpha)}}\ketbra{0}{1},\\
&\Lambda^\alpha\left(\ketbra{1}{0}\right) = \frac{\alpha}{\sqrt{\alpha(2-\alpha)}}\ketbra{1}{0},\\&\Lambda^\alpha\left(\ketbra{1}{1}\right) = \ketbra{1}{1},
\end{align*}
such that $\mathbf{I}_{A_1}\otimes\Lambda^\alpha_{A_2}(\psi^\alpha_{\bf A})=\rho^\alpha_{\bf A}$, where $\psi^{\alpha}_{\bf A}=\ket{\psi^\alpha}_{\bf A}\bra{\psi^\alpha}$.
\end{proposition}

\begin{figure}[tbh!]
\centering
\includegraphics[width=0.489\textwidth]{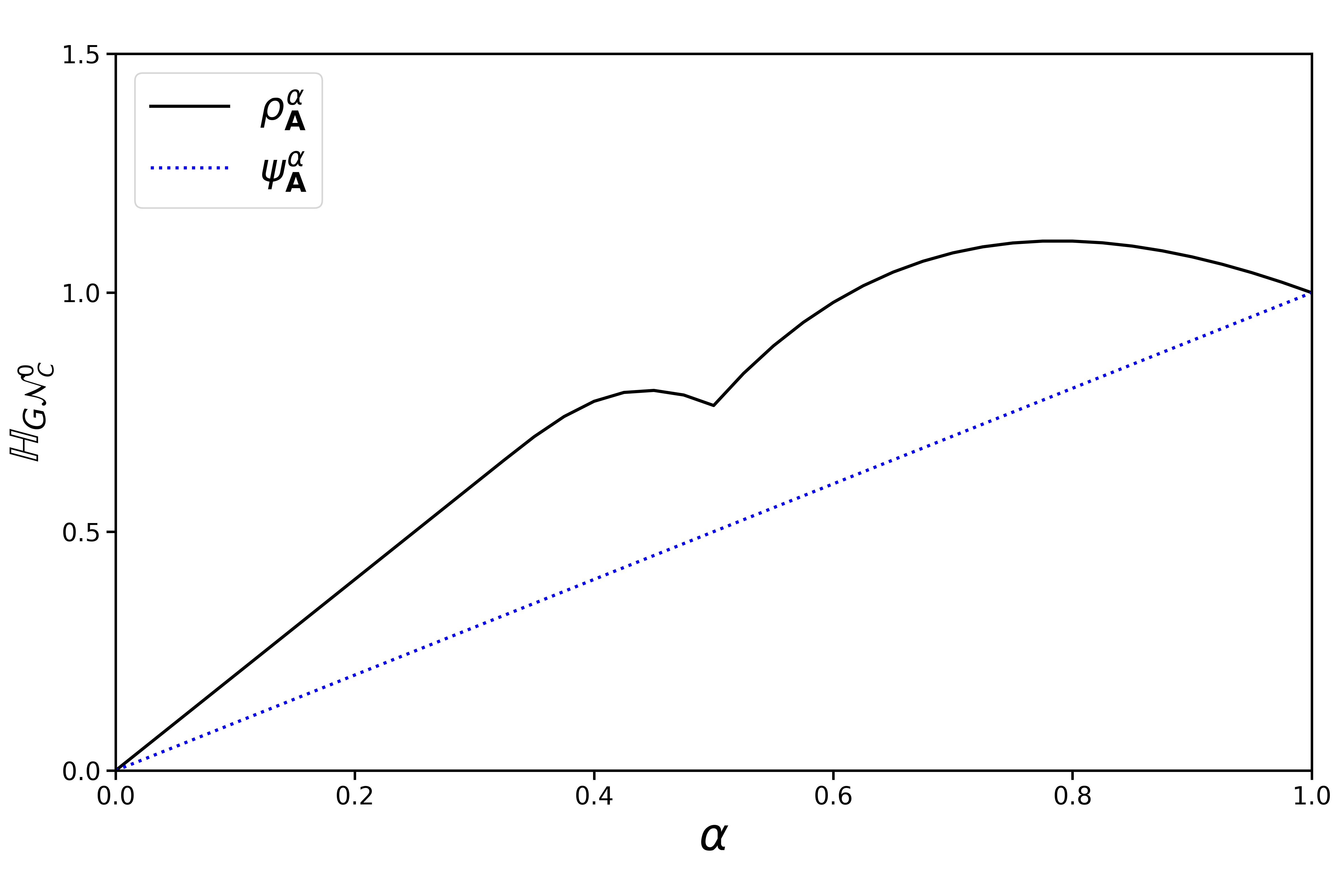}
\caption{The set of free operation of quantum discord and ergodiscord can not be the same. In the figure ergodiscord for two states $\psi^\alpha$ and $\rho_{\bf A}^\alpha$ has been plotted for  Hamiltonian $h_{\bf A}=h_{A_1}\otimes\mathbf{I}_{A_2}+\mathbf{I}_{A_1}\otimes h_{A_2}$, with $h_{A_1}:=6\ket{1}\bra{1}~\&~h_{A_2}:=2\ket{1}\bra{1}$. The solid black curve represents $\mathbb{H}_{G\mathcal{N}^0_{C}}(\rho^{\alpha}_{\bf A})$ and the blue dotted curve corresponds to $\mathbb{H}_{G\mathcal{N}^0_{C}}(\psi^{\alpha}_{\bf A})$. Proposition \textbf{3} states that $\mathbf{I}_{A_1}\otimes\Lambda^\alpha_{A_2}(\psi^\alpha_{\bf A})=\rho^\alpha_{\bf A}$. From the figure it is easy to see that $\mathbb{H}_{G\mathcal{N}^0_{C}}(\rho^{\alpha}_{\bf A})\geq\mathbb{H}_{G\mathcal{N}^0_{C}}(\psi^{\alpha}_{\bf A})~\forall~\alpha$. Hence, local CPTP on $A_2$ part can increase the value of nonlocal energy locking from $A_1\mapsto A_2$. Therefore, unlike quantum discord, not all local CPTP on $A_2$'s part qualify as free operation of ergodiscord.}
\label{figs3}
\end{figure}
\begin{proof}
Now, from Fig.~\ref{figs3}, one can easily see that $\mathbb{H}_{G\mathcal{N}^0_{C}}(\rho^{\alpha}_{\bf A})\geq\mathbb{H}_{G\mathcal{N}^0_{C}}(\psi^{\alpha}_{\bf A})~\forall~\alpha$. It simply implies that the class of CPTP maps $\Lambda^\alpha$ can not be a free operation when one deals with ergodiscord from $A_1\mapsto A_2$. Whereas, in case of quantum discord from $A_1\mapsto A_2$ all CPTP maps on $A_2$ part, including all the $\Lambda^\alpha$'s, are free. This establishes the claim.
\end{proof}
Few observations as listed in the following remarks are in order:
\begin{remark}
Local CPTP on $A_2$-part, however, can not increase energy locking from zero to non-zero. This follows from the fact that zero energy locking implies that the bipartite state is of the classical-quantum(CQ) form. If $\rho_{\bf A}$ be such CQ state, local CPTP on $A_2$-part can not make it a quantum-quantum or quantum-classical state. In other words, the set of all such CQ states is closed under any CPTP map on $A_2$ subsystem. This further highlights the faithfulness of ergodiscord.   
\end{remark}
\begin{remark}
The ergodiscord of a bipartite system increases under local CPTP operations on \( A_2 \); however, unlike the geometric discord \cite{Dakic2010} based on Hilbert-Schmidt norm, this increase is inherently bounded. Specifically, the ergodiscord is always upper-bounded by the total energy of the system's Hamiltonian. For instance, consider a two-qubit state \( \rho_{\bf A} \) with Hamiltonians \( h_{A_1} = \epsilon_1 \ket{1}\bra{1} \) and \( h_{A_2} = \varepsilon_2 \ket{1}\bra{1} \). After applying a local CPTP operation on subsystem \( A_2 \), the ergodiscord remains trivially bounded above by \( \epsilon_1 + \epsilon_2 \). This avoids the unphysical behavior exhibited by geometric discord, whose increase is unbounded \cite{Piani2012}. Moreover, while geometric discord is purely a mathematical construct with no clear physical interpretation, ergodiscord has a well-defined physical meaning that can be experimentally verified. 
\end{remark}
\begin{remark}
The resource-theoretic formulations of ergodiscord and quantum discord are fundamentally different. Although both share the same set of free states, their free operations differ due to the distinct operational and physical contexts they address. A parallel can be drawn to the resource theories of coherence and asymmetry, which also share the same set of free states but differ in their free operations \cite{Bartlett2007,Winter2016,Piani2016}.
\end{remark}
\begin{remark}
One may argue that if the system Hamiltonian changes, the value of ergodiscord will also change for the same system. This can be understood as follows: although ergodiscord quantifies correlations, it implicitly depends on the observable (i.e., energy), and consequently, its value varies with the Hamiltonian. As a physically motivated quantity, ergodiscord is a byproduct of both the quantum state and the Hamiltonian. To address this arbitrariness, one might normalize ergodiscord by dividing it by the maximum energy eigenvalue minus the minimum energy eigenvalue of the Hamiltonian ($[\epsilon^A_{\max}+\epsilon^B_{\max}]-[\epsilon^A_{\min}+\epsilon^B_{\min}]$). This normalization ensures that the quantity remains zero for non-discordant states and finite for arbitrary Hamiltonians in the presence of discord. From a physical perspective, this normalized measure would represent the ratio of energy locked in correlations to the maximum ergotropy.
\end{remark}

\section{Conclusion}
In this work, we build an elegant operational approach to study nonclassical correlations. Our approach relies on the different probing method one can apply on a system consisting of specially separated subsystems. The success of our approach lies in the realization that when specially separated systems are correlated, that very correlation can lock away the values of observable quantities from an observer who is constrained to restricted forms of probing mechanisms. The power of this simple fact is exhibited in Observation~\ref{obs1}. However, the scenario becomes a bit tricky once one aims to study nonclassical correlations. One needs to properly define what `classical probing' means in the context of our approach, which mathematically is nothing but a subset of classical channels shared between the spatially separated sub-parts. Correctly identifying the subset relies on the observable one wants to measure (please refer to Fig.~\ref{figs5} in the appendix). If a resource theoretic formalism of the same is available, the job is relatively simpler. However, we go on to show the generality of our approach by building on the framework for observables (ergotropy) which, to the best of our knowledge, does not have any resource theory of its own. We term this nonclassicality as ergodiscord and observed an intriguing phenomenon called nonlocal locking of energy.

The phenomenon of energy locking, apart from its foundational importance, highlights the nontrivial utility of nonclassical correlations in designing efficient quantum batteries \cite{Ferraro2018,Andolina2019,Yang2023}. In multi-cell quantum batteries, stored energy can be lost due to factors such as decoherence and interactions, both local and global. For isolated systems evolving under unitary processes, energy drainage occurs solely through local interactions when inter-system interactions are absent. Our findings demonstrate that classically correlated states cannot store energy when only local interactions and classical communication channels are available between pairs of isolated quantum batteries. Theorem~\ref{theo2} establishes that for systems with non-degenerate Hamiltonians, nonclassical correlations quantified by ergodiscord enable efficient energy storage under these constraints. Moreover, Theorem~\ref{theo4} reveals the intriguing result that mixed states can store more energy than the maximally entangled state of the same system. Importantly, the notion of ergodiscord is different than the concept proposed in \cite{Zurek2003}, which in literature is also known as thermal discord \cite{Jimnez2019}. Recent works, such as Refs.~\cite{Francica2017,Bernards2019}, have also examined the role of nonclassical correlations in ergotropic work extraction. However, the measures used in those studies are not faithful indicators of nonclassical correlations. In contrast, ergodiscord provides a reliable measure with clear operational significance. While our analysis primarily focuses on bipartite systems, the generality of our framework opens avenues for exploring energy storage in multipartite quantum systems.

\begin{acknowledgements}
MA acknowledges the funding supported by the European Union  (Project QURES- GA No. 101153001). SRC acknowledges support from University Grants Commission, India (reference no. 211610113404). GA acknowledges support by the UK Research and Innovation (UKRI) under EPSRC Grant No.~EP/X010929/1. MB acknowledges funding from the National Mission in Interdisciplinary Cyber-Physical systems from the Department of Science and Technology through the I-HUB Quantum Technology Foundation (Grant no: I-HUB/PDF/2021-22/008).
\end{acknowledgements}
\onecolumngrid
\section*{Appendix I: More On Observation \textbf{$1$}}
A noisy operation $\phi_\gamma^{no}:\mathcal{D}(\mathbb{C}^{d_\gamma})\mapsto\mathcal{D}(\mathbb{C}^{d_\gamma})$ is defined as
\begin{align}
\phi_\gamma^{no}(\rho_\gamma):=\Tr_{z}\left[U_{kz}\left(\rho_\gamma\otimes\frac{\mathbf{I}_z}{d_z}\right)U^\dagger_{kz}\right].   
\end{align}
Consider now the function  $f_\gamma(\rho_\gamma)\equiv f_\gamma^{pe}(\rho_\gamma):=\log d_\gamma-S(\rho_\gamma)$ which satisfies the following properties \cite{Horodecki2003(1)}:
\begin{itemize}
\item[(p1)] $0\le f^{pe}(\rho_\gamma)\le\log d_\gamma$ for a state $\rho_\gamma\in\mathcal{D}(\mathbb{C}^{d_\gamma})$;
\item[(p2)] $f^{pe}$ remains invariant under unitaries, {\it i.e.}, $f^{pe}\left(u_\gamma\rho u_\gamma^\dagger\right)=f^{pe}(\rho_\gamma)$;  
\item[(p3)] $f^{pe}$ is non-increasing under noisy operations, {\it i.e.}, $f^{pe}\left(\phi_\gamma(\rho)\right)\le f^{pe}(\rho_\gamma),~\forall~\phi_\gamma\in\Phi^{no}_\gamma$.
\end{itemize} 
The variation of the function $f^{pe}$ over the set of noisy operations $\Phi^{no}$ thus becomes 
\begin{align}
\Delta (f^{pe})_\gamma^{\Phi^{no}_\gamma}(\rho_\gamma)&:=\underset{\phi_\gamma\in\Phi^{no}_\gamma}{\mbox{max}}\left[f^{pe}_\gamma\left(\rho_\gamma\right)-f^{pe}_\gamma\left(\phi_\gamma\left(\rho_\gamma\right)\right)\right]\nonumber\\
&=f^{pe}_\gamma(\rho_\gamma)-\underset{\phi_\gamma\in\Phi^{no}_\gamma}{\mbox{min}}\left[f^{pe}_\gamma\left(\phi_\gamma\left(\rho_\gamma\right)\right)\right]\nonumber\\
&= f^{pe}_\gamma(\rho_\gamma). \label{delmax}  
\end{align}
The last line follows from the fact that under noisy operation the minimum value of $f^{pe}_\gamma\left(\phi_\gamma\left(\rho_\gamma\right)\right)$ can be achieved to zero by evolving the state $\rho_\gamma$ to the completely mixed state $\mathbf{I}_\gamma/{d_\gamma}$. Now, for the bipartite case $({\bf A}\equiv A_1A_2)$ we have
\begin{align}
\mathbb{F}_{GL}^{pe}(\rho_{\bf A})&= f^{pe}_{\bf A}\left(\rho_{\bf A}\right)-f^{pe}(\rho_1)-f^{pe}_2(\rho_2)\nonumber\\
&=\log(d_1d_2)-S(\rho_{\bf A})-\log d_1+S(\rho_1)-\log d_2+S(\rho_2)\nonumber\\
&=S(\rho_1)+S(\rho_2)-S(\rho_{\bf A})=I(A_1:A_2).
\end{align}
This establishes the claim in Observation \textbf{$1$}.

\section*{Appendix II: Proof of Proposition \textbf{$1$}}
\begin{proof}
Let $\vec p = (p_1,p_2,\cdots p_d)$ and $\vec q =  (q_1,q_2,\cdots q_d)$ are two probability vectors with their elements arranged in increasing order. The vector $\vec p$ majorizes $\vec q$, denoted as $\vec p\succ\vec q$, if and only if \cite{Marshall2011} 
\begin{align}\label{propeq2}
\sum_{i\in[l]} \left(p_i-q_i\right) \leq 0,~\forall~l\in[d-1],~~\&~~
\sum_{i\in [d]} \left(p_i-q_i\right) = 0~.   
\end{align} 
Consider a density operator $\rho\in\mathcal{D}(\mathbb{C}^{d})$ and a unital CPTP map $\Lambda:\mathcal{D}(\mathbb{C}^{d})\mapsto\mathcal{D}(\mathbb{C}^{d})$, with $\vec p $ and $\vec q $ denoting eigenvalues (arranged in increasing order) of $\rho$ and $\Lambda(\rho)$, respectively. Clearly, the vector $\vec p$ majorizes $\vec q$ \cite{Gour2015}. 
Whenever $\vec p\succ\vec q$, the majorization conditions ensure that $-\infty<p_i\log q_i\le0,~\forall~i\in[d]$. Thus we have,
\begin{align}\label{propeq3}
\sum_{i\in[d]}(q_i-p_i)\log q_i&=\sum_{i=1}^{d}(q_i-p_i)\log q_i\nonumber\\
&= \left(q_1-p_1\right)\log q_1+\left(q_2-p_2\right)\log q_2+\cdots+\left(q_d-p_d\right)\log q_d\nonumber \\
&=\left(q_1-p_1\right)\log q_1+\Big((q_1+q_2)-(p_1+p_2)\Big)\log q_2+\cdots+\left(\sum_{i\in[d]}q_i-\sum_{i\in[d]}p_i\right)\log q_d \nonumber \\
&~~-\left\lbrace (q_1-p_1)\log q_2+\Big((q_1+q_2)-(p_1+p_2)\Big)\log q_3+\cdots+\left(\sum_{i\in[d-1]}q_i-\sum_{i\in[d-1]}p_i\right)\log q_d\right\rbrace\nonumber\\
&=\left(q_1-p_1\right)\left(\log q_1-\log q_2\right)+\Big((q_1+q_2)-(p_1+p_2)\Big)\left(\log q_2-\log q_3\right)+\cdots\nonumber\\
&~~~~~~~~~~~+\left(\sum_{i\in[d-1]}q_i-\sum_{i\in[d-1]}p_i\right)\left(\log q_{d-1}-\log q_d\right)+\left(\sum_{i\in[d]}q_i-\sum_{i\in[d]}p_i\right)\log q_d\nonumber \\
& = \sum_{l=2}^d\left\{\sum_{i\in[l-1]}\left(q_i-p_i\right)\right\}\Big(\log q_{l-1}-\log q_{l}\Big)+\left(\sum_{i\in[d]}(q_i-p_i)\right)\log q_d.
\end{align}
Now using Eq.(\ref{propeq2}) and the fact that $\log$ is a monotonically increasing function, from Eq.(\ref{propeq3}) we have,
\begin{align}\label{propeq4}
\sum_{i\in[d]}^{d}(q_i-p_i)\log q_i\le0.
\end{align}
Now, suppose that $S(\Lambda(\rho))=S(\rho)$, i.e., $\sum_{i\in[d]} q_i\log q_i=\sum_{i\in[d]} p_i\log p_i$. Then again using the fact $-\infty<p_i\log q_i\le0,~\forall~i\in[d]$, we have,
\begin{align}\label{propeq1}
\sum_{i\in[d]} q_i\log q_i-\sum_{i\in[d]} p_i\log q_i&=\sum_{i\in[d]} p_i\log p_i-\sum_{i\in[d]} p_i\log q_i,\nonumber\\
\Rightarrow \sum_{i\in[d]} \left(q_i-p_i\right)\log q_i &= \sum_{i\in[d]} p_i\left(\log p_i-\log q_i\right),\nonumber\\
\Rightarrow \sum_{i\in[d]} \left(q_i-p_i\right)\log q_i &=H\left(\vec p || \vec q\right)\geq 0.
\end{align}
The last inequality follows from the non negativity of relative entropy $H(\vec{p}||\vec{q})$. Now comparing Eq.(\ref{propeq4}) and Eq.(\ref{propeq1}), we have 
\begin{align}
H\left(\vec p || \vec q\right)=0\Rightarrow  \vec p = \vec q.
\end{align}
This completes the proof. 
\end{proof}

\section*{Appendix III: Proof of Theorem \textbf{$1$}}
Before proving Theorem $1$, in the following we first characterize the the set of free channel $\mathcal{N}^0_C$ through which $A_1$ sub-part is send to $A_2$. 

A quantum channel $N:\mathcal{D}(\mathbb{C}^d)\mapsto\mathcal{D}(\mathbb{C}^d)$ is called a measure-and-prepare channel if it is of the form 
\begin{align}
N(\rho)=\sum_{i\in[n]}\tr(\pi_i\rho)\ket{\psi_i}\bra{\psi_i},    
\end{align}
where the $\pi_i\ge0~\&~\sum_{i\in[n]}\pi_i=\mathbf{I}_d$, where $n$ can take any integer value. Notably, such a channel breaks entanglement and hence has zero quantum capacity \cite{Horodecki2003}. Even though, theses channels can not send any quantum information we will not consider them to be classical channels. A channel $N$ is called classical channel if it carries information of distinct classical symbols only. Set of such channels $\mathcal{N}_C$ are proper subset of measure-and-prepare channel and takes the form   
\begin{align}\label{cc1}
N(\rho)=\sum_{i\in[n]}\tr(\pi_i\rho)\ket{\psi_i}\bra{\psi_i},    
\end{align}
where $n\le d$ and $\langle\psi_i|\psi_{i^\prime}\rangle=\delta_{ii^\prime}$. However, for the choices $(f^{pe}_\gamma,\Phi^{no}_\gamma)$ not all $\mathcal{N}_C$ are free. In the following lemma we characterize the set of free channel $\mathcal{N}^0_C\subset\mathcal{N}_C$.
\begin{manuallemma}{S$1$}\label{lemma1} 
In resource theory of purity a free classical channel $N\in\mathcal{N}^0_C$ takes the form 
\begin{align}
N(\rho)=\sum_{i\in[d]}\tr(\pi_i\rho)\ket{\psi_i}\bra{\psi_i},   
\end{align}
where $\langle\psi_i|\psi_{i^\prime}\rangle=\delta_{ii^\prime}$ and $\tr(\pi_i)=1,~\forall~i\in[d]$.
\end{manuallemma}
\begin{proof}
Recall that $S(\Lambda^u(\rho))\ge S(\rho)$ for a unital CPTP map $\Lambda^u$, and hence $f^{pe}(\Lambda^u(\rho))\le f^{pe}(\rho)$. On the other hand, for any non-unital CPTP map $\Lambda^{nu}$, we have $f^{pe}(\Lambda^{nu}(\mathbf{I}/d))> f^{pe}(\mathbf{I}/d)=0$, and therefore they are not free. Thus any free classical channel $N\in\mathcal{N}^0_C$ must satisfy $N(\mathbf{I}/d)=\mathbf{I}/d$, which yields the form of $N$ as in Lemma.     
\end{proof}
\subsection*{Proof of Theorem \textbf{$1$}}
\begin{proof} (Theorem \textbf{$1$})\\  
{\bf if part:} Here we will prove that for every CQ state 
\begin{align}\label{prop1eq1}
\rho_{\bf A}=\sum p_{i_1i_2}\psi^{i_1}_1\otimes\psi^{i_2}_2,~ \mbox{with}~ \langle\psi^{i_1}_1|\psi^{i^\prime_1}_1\rangle=\delta_{i_ii^\prime_1},   
\end{align}
we have $\mathbb{F}^{pe}_{G\mathcal{N}^0_C}(\rho_{\bf A})=0$. 

As have already discussed in the main text 
\begin{align}
\Delta f^{\{\Phi^{no}_\gamma\}}_{\mathcal{N}_C^0}(\rho_{\bf A})\leq \Delta f^{\Phi^{no}_{\bf A}}_{\bf A}(\rho_{\bf A}),~~\forall~ \rho_{\bf A}\in\mathcal{D}(\mathbb{C}^{d_1}\otimes\mathbb{C}^{d_2}).   
\end{align}
Also recall that, evaluating $\Delta f^{\{\Phi^{no}_\gamma\}}_{\mathcal{N}_C^0}(\rho_{\bf A})$ essentially contains two key steps: first applying $\phi_1\in\Phi^{no}_1$ in $A_1$'s local lab, and then sending the $A_1$ sub-part to the $A_2$ party through the channel $N \in \mathcal{N}_C^0$. The channel $N$ has the following form (follows from Lemma \ref{lemma1})
\begin{align} \label{prop1eq2}
&N(\rho_1)=\sum_{l\in[d_1]}\tr(\pi_l\rho_1)\ket{\alpha^l}_1\bra{\alpha^l},~\mbox{with}~\langle\alpha^l|\alpha^{l^\prime}\rangle=\delta_{ll^\prime},\\  
&\mbox{and}~\{\pi_l\}~\mbox{is a POVM with}~\tr(\pi_l)=1,~\forall~l\in[d_1].\nonumber
    \end{align}
Given a CQ state $\rho_{\bf A}$ of the form of Eq.(\ref{prop1eq1}), let no $\phi_1$ is applied on the $A_1$ part, rather part is sent through the free channel 
\begin{align}
N^\star(\rho_1)=\sum_{l\in[d_1]}\tr(\psi^{l}\rho_1)\ket{\psi^{l}}_1\bra{\psi^{l}},
\end{align}
where $\{\ket{\psi^l}_1\}_{l=1}^{d_1}$ are the states appearing in the $A_1$ part of expression Eq.(\ref{prop1eq1}). Since $N^\star\otimes\mathcal{I}(\rho_{\bf A})= \rho_{\bf A}$, therefore the CQ state $\rho_{\bf A}$ will reach in $A_2$'s lab completely unchanged; here $\mathcal{I}$ denotes the identity map. Thus, the extracted value of $\Delta f^{\{\Phi^{no}_\gamma\}}_{\mathcal{N}_C^0}(\rho_{\bf A})$ achieves it maximum value $\Delta f^{\Phi^{no}_{\bf A}}_{\bf A}(\rho_{\bf A})$, implying $\mathbb{F}^{pe}_{G\mathcal{N}^0_C}(\rho_{\bf A})=0$. This completes the {\it if part} of the proof. \\\\
{\bf only if part:}  Starting with an arbitrary state $\chi_{\bf A}$ we will now show that
\begin{align}\label{oi1}
\mathbb{F}^{pe}_{G\mathcal{N}^0_C}(\chi_{\bf A}):=\underbrace{\Delta \left(f^{pe}\right)_{\bf A}^{\Phi_{\bf A}}(\chi_{\bf A})}_{1^{st}~term}-\underbrace{\Delta \left(f^{pe}\right)^{\{\Phi_\gamma\}}_{\mathcal{N}_C^0}(\chi_{\bf A})}_{2^{nd}~term}=0~\implies~ \chi_{\bf A}~ \mbox{is a CQ state}. 
\end{align}
Recall that a channel $\Lambda:\mathcal{D}(C^{d})\to\mathcal{D}(C^{d})$ is called unital if it maps completely mixed state to itself, {\it i.e.}, $\Lambda\left(\mathbf{I}/d\right)=\mathbf{I}/d$. Set of all such channels is denoted by $\mathcal{U}(d)$. Since for a $\Lambda\in\mathcal{U}(d)$ we have $f^{pe}(\Lambda(\rho))\le f^{pe}(\rho),~\forall\rho\in\mathcal{D}(C^{d})$, theses channels are considered to be free. On the other hand, as shown in Lemma \ref{lemma1}, a free classical channel $N\in\mathcal{N}^0_C$ reads as
\begin{align}\label{oi2}
&N(\rho_1):=\sum_{l\in[d_1]}\tr(\pi_l\rho_1)\ket{\alpha^l_1}\bra{\alpha^l_1},~\mbox{where}~ \langle\alpha^l_1|\alpha^{l^\prime}_1\rangle=\delta_{ll^\prime},\\
&\mbox{and}~\{\pi_l\}~\mbox{is a POVM with}~\tr(\pi_l)=1,~\forall~l\in[d_1].\nonumber
\end{align}
Manifestly it follows that $\mathcal{N}^0_C\subset\mathcal{U}(\mathbb{C}^{d_1})$. Now, the "$1^{st}~term$" in Eq.(\ref{oi1}) simply boils down to 
\begin{align}\label{oi5}
\Delta \left(f^{pe}\right)_{\bf A}^{\Phi_{\bf A}}(\chi_{\bf A})=\log(d_1d_2)-S(\chi_{\bf A}).   
\end{align}
For the "$2^{nd}~term$" let us recall its expression:
\begin{align}\label{oi3}
\Delta \left(f^{pe}\right)^{\{\Phi_\gamma\}}_{\mathcal{N}_C^0}(\chi_{\bf A})=\max_{\left\{\phi_1\in\Phi^{no}_1,~\phi_2\in\Phi^{no}_2,~\phi_{\bf A}\in\Phi^{no}_{\bf A},~~N\in\mathcal{N}^0_C\right\}}\left[\Delta f^{\phi_1}_1(\chi_1)+\Delta f^{\phi_2}_2(\chi_2)+\Delta f^{\phi_{\bf A}}_{\bf A}(N(\sigma_{\bf A}))-\mathcal{C}(N)\right],   
\end{align}
where $\sigma_{\bf A}=(\phi_1\otimes\phi_2)(\chi_{\bf A})$. Clearly, $C(N)=0$ as $N\in\mathcal{N}^0_C$. Note that the marginal $\chi_1$ being completely mixed, makes the optimization over $\Phi^{no}_1$ irrelevant. Furthermore, the action of $\phi_2$ can be absorbed in $\phi_{\bf A}$, which in turn reduces the optimization over $\Phi^{no}_2$. Accordingly, Eq.(\ref{oi3}) modified as
\begin{align}\label{oi4}
\Delta \left(f^{pe}\right)^{\{\Phi_\gamma\}}_{\mathcal{N}_C^0}(\chi_{\bf A})=\max_{\left\{\phi_{\bf A}\in\Phi^{no}_{\bf A},~~ N\in\mathcal{N}^0_C\right\}}\left[\Delta f^{\phi_{\bf A}}_{\bf A}(N(\chi_{\bf A}))\right].   
\end{align}
From Eqs.(\ref{oi1}), (\ref{oi5}), and (\ref{oi4}), we therefore have
\begin{align}
\max_{\left\{\phi_{\bf A}\in\Phi^{no}_{\bf A},~~ N\in\mathcal{N}^0_C\right\}}\left[\Delta f^{\phi_{\bf A}}_{\bf A}(N(\chi_{\bf A}))\right]=\log(d_1d_2)-S(\chi_{\bf A}).
\end{align}
Suppose the optimization over $\mathcal{N}^0_C$ is achieved for some $N^\star$. Then we have 
\begin{align}
&\max_{\left\{\phi_{\bf A}\in\Phi^{no}_{\bf A}\right\}}\left[\Delta f^{\phi_{\bf A}}_{\bf A}(N^\star(\chi_{\bf A}))\right]=\log(d_1d_2)-S(\chi_{\bf A})\nonumber\\
&\Rightarrow S((N_1^\star\otimes\mathcal{I}_2)(\chi_{\bf A}))=S(\chi_{\bf A}).
\end{align}
Using Proposition \textbf{$1$}, we thus have
\begin{align}
\chi_{\bf A}=(N_1^\star\otimes\mathcal{I}_2)(\chi_{\bf A})=\sum_{l\in[d_1]}\tr((\pi^\star_l\otimes\mathbf{I})\chi_{\bf A})\ket{\alpha^l}_1\bra{\alpha^l}\otimes \chi_{A_2|l}.
\end{align}
Here, $\chi_{A_2|l}:=\tr_{1}\left[\left(\pi^\star_l\otimes\mathbf{I}\right)\chi_{\bf A}\right]/\tr\left[\left((\pi^\star_l\otimes\mathbf{I})\right)\chi_{\bf A}\right]$ represents the normalized state of subpart $A_2$, conditioned on the positive operator $\{\pi^\star_l\}$ acting on part $A_1$. This ensures that $\chi_{\bf A}$ conforms to the CQ form. This completes {\it only if part} of the proof, and hence the Theorem \textbf{$1$} is proved. 
\end{proof}

\section*{Appendix IV: Proof of Theorem \textbf{$2$}}
We first present few Lemmas that will be needed for proving Theorem \textbf{$2$}. We start by characterizing the relevant set of zero-cost classical communication channels $\mathcal{N}^0_C$. Let us recall that a free classical channel $N$ in resource theory of purity satisfies two conditions: (i) $f^{pe} (N(\rho)) \leq f^{pe} (\rho),~\forall~\rho$, and (ii) $\Delta f^{{\Phi^{no}}} (N(\rho)) \leq \Delta f^{{\Phi^{no}}} (\rho),~\forall~\rho$.  It is worth noting that both of these conditions solely depend on the entropy of the given state, and moreover one implies the other. However, within the observable approach, the relationship between the expectation value of an observable $\langle \mathcal{O}\rangle_{(\rho)}$ and its variation $\Delta \mathcal{O}^U (\rho)$ exhibit more complex behaviour. Hence, we demand a free classical channel $N$ to fulfill both the following two criteria: \begin{align*}
&\mbox{\bf (C1)}~\tr(\mathcal{O}N(\rho))\le\tr(\mathcal{O}\rho),~\forall~\rho\in\mathcal{D}(\mathbb{C}^d);\\  
&\mbox{\bf (C2)}~\Delta \mathcal{O}^U (N(\rho))\le \Delta \mathcal{O}^U (\rho), ~ \forall\rho\in\mathcal{D}(\mathbb{C}^{d}).
\end{align*}
In the following lemma we completely characterize the set of free channels that satisfies these two criteria.  
\begin{manuallemma}{S$2$}\label{lemma2}
Given a non-degenerate $\mathcal{O}\in$ Herm$(\mathbb{C}^{d})$, a classical channel $N$ satisfying the criteria (C1) and (C2) reads as
\begin{subequations}
\begin{align}
either~~N\left(\rho\right) &= \ket{o_1}\bra{o_1},\label{l2a}\\
or~~N\left(\rho\right) &= \sum_{l\in[d]}\tr\left(\ket{o_l}\bra{o_l}\rho\right) \ket{o_l}\bra{o_l},\label{l2b}
\end{align}   
\end{subequations}
where, $\{\ket{o_l}\}_{l\in[d]}$ are the eigenstates of $\mathcal{O}$ and forms a complete basis of $\mathbb{C}^d$.
\end{manuallemma}
\begin{proof}
Recalling from Eq.(\ref{cc1}), a classical channel reads as
\begin{align}\label{l21}
N(\rho)=\sum_{l\in[n]}\tr(\pi_l\rho)\ket{\psi_l}\bra{\psi_l},    
\end{align}
where $\langle\psi_l|\psi_{l^\prime}\rangle=\delta_{ll^\prime}$, $\pi_l\geq0~\&~\sum_{l\in[n]} \pi_l=\mathbf{I}_d$ with $n\le d$.
Let $\mathcal{O}:=\sum_{l\in[d]}\lambda_l\ketbra{o_l}{o_l}$ be the spectral representation of $\mathcal{O}$ with $\lambda_l$'s being arranged in increasing order, i.e, $\lambda_1\leq\lambda_2\leq\cdots\leq\lambda_d$. 

Consider a classical channel as of Eq.(\ref{l21}) with $n=1$ which represents the pin map $N(\rho) = \ket{\psi_1}\bra{\psi_1}$ with $\pi_1 = \mathbf{I}_d$. When the channel is fed with the state $\ket{o_1}$, the eigenstate of $\mathcal{O}$ corresponding to the lowest eigenvalue, the condition (C1) will be satisfied only if $\ket{\psi_1}=\ket{o_1}$. Moreover, this choice also turns out to be compatible with the requirement of condition (C2). Accordingly, the free channel takes the form of Eq.(\ref{l2a}).

Consider, now the case $1 < n \leq d$. Feeding the state $\rho=\ket{o_1}\bra{o_1}$ into the channel and imposing (C1) we obtain the following conclusions: $\pi_1=\ket{o_1}\bra{o_1}$, $\ket{\psi_1}=\ket{o_1}$, and $\tr(\pi_l\pi_1)=\delta_{l1},~\forall~l \in [n]$. Accordingly, the channel takes the form:
\begin{align}
N(\rho)= \tr(\ket{o_1}\bra{o_1}\rho)\ket{o_1}\bra{o_1}+\sum_{l=2}^{n}\tr(\pi_l\rho)\ket{\psi_l}\bra{\psi_l}.    
\end{align}
Moving forward, by setting $\rho=\ket{o_2}\bra{o_2}$ and imposing condition (C1), we obtain $\pi_2=\ket{o_2}\bra{o_2}$, $\ket{\psi_2}=\ket{o_2}$, and $\tr(\pi_l\pi_2)=\delta_{l2},~\forall~l \in [n]$, which thus lead to
\begin{align}
N(\rho)= \tr(\ket{o_1}\bra{o_1}\rho)\ket{o_1}\bra{o_1}+\tr(\ket{o_2}\bra{o_2}\rho)\ket{o_2}\bra{o_2}+\sum_{l=3}^{n}\tr(\pi_l\rho)\ket{\psi_l}\bra{\psi_l}.    
\end{align}
Continuing this we extend the argument for $l\in [n-1]$ and obtain
\begin{align}\label{l22}
N(\rho)= \sum_{l \in [n-1]}\tr(\ket{o_l}\bra{o_l}\rho)\ket{o_l}\bra{o_l}+\tr(\pi_n\rho)\ket{\psi_n}\bra{\psi_n}.    
\end{align}
Consider now the completely mixed state $\rho=\mathbf{I}/d$ for which the variation $\Delta\mathcal{O}^U(\mathbf{I}/d)=0$. The action of the channel (\ref{l22}) on this state yields
\begin{align}
N(\mathbf{I}/d)= \frac{1}{d}\sum_{l \in [n-1]}\ket{o_l}\bra{o_l}+\left(1-\frac{n-1}{d}\right)\ket{\psi_n}\bra{\psi_n}.    
\end{align}

\begin{align*}
\left(1-\frac{n-1}{d}\right)&\leq \frac{1}{d} \Rightarrow d \leq n ,\\
\&~\ket{\psi_n}\bra{\psi_n}&=\ket{o_n}\bra{o_n}.
\end{align*}
Since for the classical in Eq.(\ref{l21}) we have $n \leq d$, thus we conclude $n=d$. Accordingly a free classical channel $N\in\mathcal{N}^0_C$ takes the form of Eq.(\ref{l2b}). This concludes the proof. 
\end{proof}
\begin{figure}[t!]
\centering
\includegraphics[scale=0.45]{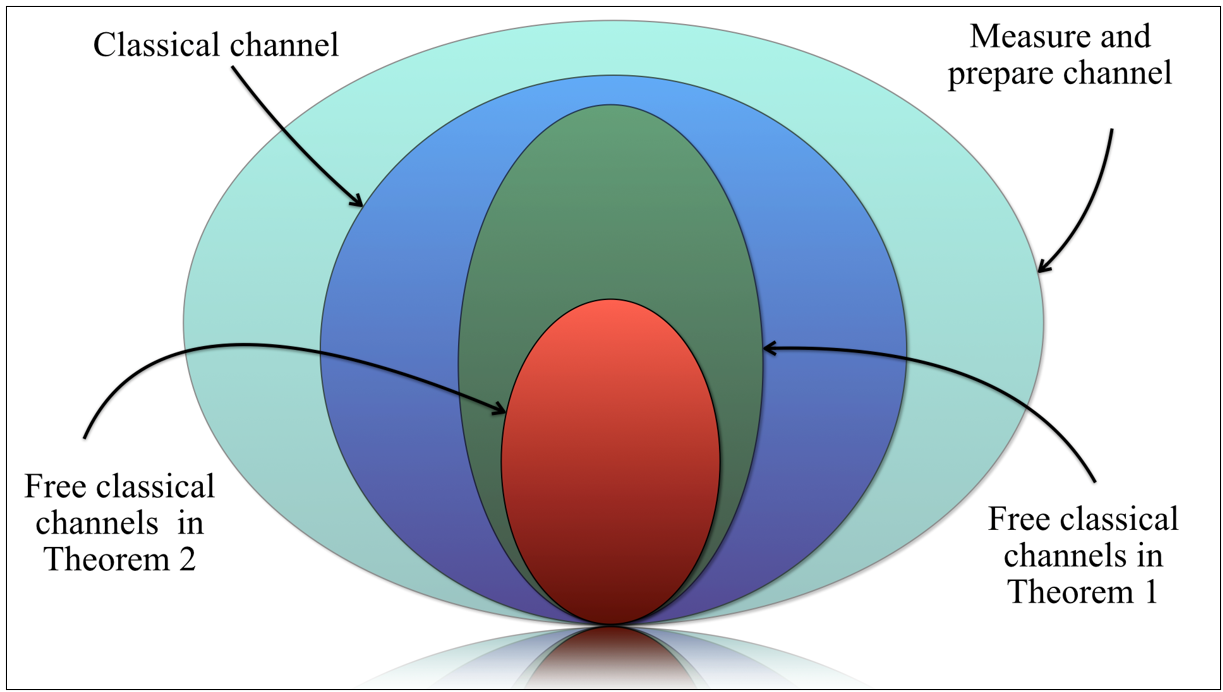}
\caption{A measure-and-prepare channel $N_{MP}:\mathcal{D}(\mathbb{C}^d)\mapsto\mathcal{D}(\mathbb{C}^d)$ reads as
$N_{MP}(\rho)=\sum_{i\in[n]}\tr(\pi_i\rho)\ket{\psi_i}\bra{\psi_i}$, where $\pi_i\ge0~\&~\sum_{i\in[n]}\pi_i=\mathbf{I}_d$, with $n$ taking integer values. Let $\mathcal{N}_{MP}$ denotes the set of all such channels. A classical channel $N_C\in\mathcal{N}_C$ takes the form  $N_C(\rho)=\sum_{i\in[n]}\tr(\pi_i\rho)\ket{\psi_i}\bra{\psi_i}$, where $n\le d$ and $\langle\psi_i|\psi_{i^\prime}\rangle=\delta_{ii^\prime}$. A free classical channel $N^{0(1)}_C\in\mathcal{N}^{0(1)}_C$ appearing in Theorem 1 reads as  $N^{0(1)}_C(\rho)=\sum_{i\in[d]}\tr(\pi_i\rho)\ket{\psi_i}\bra{\psi_i}$, where $\langle\psi_i|\psi_{i^\prime}\rangle=\delta_{ii^\prime}$ and $\tr(\pi_i)=1,~\forall~i\in[d]$ (see Lemma 1). For a non-degenerate $\mathcal{O}\in$ Herm$(\mathbb{C}^{d})$, a free classical channel $N^{0(2)}_C\in\mathcal{N}^{0(2)}_C$ appearing in Theorem 2 takes the form either $N^{0(2)}_C\left(\rho\right)= \ket{o_1}\bra{o_1}$ or $N^{0(2)}_C\left(\rho\right)= \sum_{l\in[d]}\tr\left(\ket{o_l}\bra{o_l}\rho\right) \ket{o_l}\bra{o_l}$, where $\{\ket{o_l}\}_{l\in[d]}$ are the eigenstates of $\mathcal{O}$ (see Lemma 2). The above figure depicts the set inclusion relations, $\mathcal{N}^{0(2)}_C\subsetneq\mathcal{N}^{0(1)}_C\subsetneq\mathcal{N}_C\subsetneq\mathcal{N}_{MP}$. }
\label{figs5}\vspace{-.3cm}
\end{figure}

Given a quantum state $\rho$ and an observable $\mathcal{O}$, we are interested in the variation $\Delta\mathcal{O}^{\mathrm{U}}(\rho):=\max_{u\in\mathrm{U}}[\tr(\mathcal{O}\rho)-\tr(\mathcal{O} u\rho u^\dagger)]$. The state having same entropy as of the state $\rho$ and yielding minimum expectation value for the observable $\mathcal{O}$ is thus relevant in this regard. We call this state the $\mathcal{O}$-passive state.  
\begin{upstatement}{Definition S$1$}\label{def1}
Given an observable $\mathcal{O}\in$Herm$(\mathbb{C}^d)$ state $\rho_p^o\in \mathcal{D}(\mathbb{C}^d)$ is said to be the $\mathcal{O}$-passive state, if 
\begin{align}
\tr(\mathcal{O}\rho_p^o)\leq \tr(\mathcal{O}u\rho_p^o u^\dag),~\forall~u\in \mathrm{U}.
\end{align}
\end{upstatement}

\begin{manualtheorem}{S$1$}\label{theo0}
[Schur-Horn theorem \cite{Schur1923,Horn1954,Marshall2011}] Let $\mathcal{Z}\in$Herm$(\mathbb{C}^d)$ has spectral $\mathbb{R}^d\ni\vec \lambda\equiv(\lambda_1,\cdots,\lambda_d)$ and let $\mathbb{R}^d\ni\vec \mu\equiv(\mu_1,\cdots,\mu_d)$ be the diagonal elements of $\mathcal{Z}$ represented in some other basis, with both the vectors arranged in increasing order. Then $\vec \lambda\succ\vec \mu$.      
\end{manualtheorem}
Using this Theorem we provide a sufficient form for the $\mathcal{O}$-passive state. 
\begin{manuallemma}{S$3$}\label{lem3}
For the observable $\mathcal{O}=\sum_{l\in[d]}\lambda_l\ketbra{o_l}{o_l}$, with $\lambda_1\geq \lambda_2\geq\cdots\geq\lambda_d$, the state $\rho_p^o\in\mathcal{D}(\mathbb{C}^d)$ of the following form is always an $\mathcal{O}$-passive state:
\begin{align}\label{passive1}
\rho_p^o= \sum_{l\in[d]} r_l \ketbra{o_l}{o_l}, \mbox{  where}, ~r_1\leq r_2\leq\cdots r_d.
\end{align}
\end{manuallemma}
\begin{proof}
Let, $\rho_u:=u\rho_p^o u^\dag = \sum_{l\in[d]}r_l\ketbra{l}{l}$, for an arbitrary $u\in\mathrm{U}$. The state $\rho_u$ can also be written in $\{\ketbra{o_l}{o_{l^\prime}}\}$ basis as $\rho_u=\sum_{l,l^\prime\in[d]}s_{ll^\prime}\ketbra{o_l}{o_{l^\prime}}$. The diagonal entries $s_{ll}$ are denoted as $s_l$ for the sake of brevity. Note that the $\vec s_u\equiv(s_1,\cdots,s_d)$ is not arranged in any particular order. However, we can obtain $\rho_{uv}:=v\rho_uv^\dagger$ by apply another unitary $v$ such that diagonal entries $\vec s_{uv}\equiv(s_1^{\uparrow},\cdots,s_d^{\uparrow})$ are arranged in increasing order. Manifestly it follows that
\begin{align}\label{lem3eq01}
\tr(\mathcal{O}\rho_{uv})\le \tr(\mathcal{O}\rho_u).
\end{align}
Applying Schur-Horn theorem between $\rho_p^o$ and $\rho_{uv}$ we  can conclude $\vec r\succ\vec s_{uv}$. Now we have,
\begin{align}\label{lem3eq1}
\tr(\mathcal{O}\rho_p^o)- \tr(\mathcal{O}\rho_{uv})&=\sum_{l\in[d]}\lambda_l\left(r_l-s_l^{\uparrow}\right)\nonumber\\
&= \left(r_1-s_1^{\uparrow}\right)\lambda_1+\Big((r_1+r_2)-(s_1^{\uparrow}+s_2^{\uparrow})\Big)\lambda_2+\cdots+\left(\sum_{l\in[d]}r_l-\sum_{l\in[d]}s_l^{\uparrow}\right)\lambda_d \nonumber \\
&~~-\left\lbrace (r_1-s_1^{\uparrow})\lambda_2+\Big((r_1+r_2)-(s_1^{\uparrow}+s_2^{\uparrow})\Big)\lambda_3+\cdots+\left(\sum_{l\in[d-1]}r_l-\sum_{l\in[d-1]}s_l^{\uparrow}\right)\lambda_d\right\rbrace\nonumber\\
&=\left(r_1-s_1^{\uparrow}\right)\left(\lambda_1-\lambda_2\right)+\Big((r_1+r_2)-(s_1^{\uparrow}+s_2^{\uparrow})\Big)\left(\lambda_2-\lambda_3\right)+\cdots\nonumber\\
&~~~~~~~~~~~+\left(\sum_{l\in[d-1]}r_l-\sum_{l\in[d-1]}s_l^{\uparrow}\right)\left(\lambda_{d-1}-\lambda_d\right)+\left(\sum_{l\in[d]}r_l-\sum_{l\in[d]}s_l^{\uparrow}\right)\lambda_d\nonumber \\
& = \sum_{m=2}^d\left\{\sum_{l\in[l-1]}\left(r_l-s_l^{\uparrow}\right)\right\}\Big(\lambda_{m-1}-\lambda_{m}\Big)+\left(\sum_{l\in[d]}(r_l-s_l^{\uparrow})\right)\lambda_d\leq 0.
\end{align}
The inequality in last line follows from the fact that $\vec r\succ\vec s_{uv}$ and the fact that $\vec \lambda$ is arranged in increasing order. Now, Eq.(\ref{lem3eq01}) and Eq.(\ref{lem3eq1}) together imply $\tr(\mathcal{O}\rho_p^o)\le\tr(\mathcal{O}u\rho_p^ou^\dagger)$. This completes the proof. 

Furthermore, it is important to notice that, when $\mathcal{O}$ is non-degenerate, the form of the state in in Eq.(\ref{passive1}) is the only allowed form of the  $\mathcal{O}$-passive state.
\end{proof}
\begin{manuallemma}{S$4$}\label{lem4}
Given an observable $\mathcal{O}\in$ Herm$(\mathbb{C}^d)$ and a state $\rho\in \mathcal{D}(\mathbb{C}^d)$, we have $\Delta\mathcal{O}^{\mathrm{U}}(\rho)= \langle\mathcal{O}\rangle_\rho-\mathcal{O}^p(\rho)$, where $\mathcal{O}^p:\mathcal{D}(\mathbb{C}^d)\mapsto \mathbb{R}$ is defined as: $\mathcal{O}^p(\rho):= \tr(\mathcal{O}\rho_p^o)$, with $\rho_p^o$ being a $\mathcal{O}$-passive state having the same spectrum as $\rho$.
\end{manuallemma}
\begin{proof} 
Recalling the definition of $\Delta\mathcal{O}^{\mathrm{U}}(\rho)$, we have
\begin{align}
\Delta\mathcal{O}^{\mathrm{U}}(\rho)&=\max_{u\in\mathrm{U}}[\tr(\mathcal{O}\rho)-\tr(\mathcal{O} u\rho u^\dagger)]\nonumber\\
&=\tr(\mathcal{O}\rho)- \min_{u\in\mathrm{U}}\tr(\mathcal{O} u\rho u^\dagger)\nonumber\\
&=\tr(\mathcal{O}\rho)-\tr(\mathcal{O}\rho_p^o).
\end{align}
The last line directly follows from the Definition \ref{def1} of $\mathcal{O}$-passive state. 
\end{proof}
\begin{manuallemma}{S$5$}\label{lem5}
Given a non-degenerate $\mathcal{O}\in$ Herm$(\mathbb{C}^d)$ and a channel $N$ of the form Eq.(\ref{l2b}), if $\mathcal{O}^p(\rho)=\mathcal{O}^p(N(\rho))$, then $\rho$ is a fixed point of $N$, i.e., $N(\rho)=\rho$.
\end{manuallemma}
\begin{proof}
Let, $\vec p=(p_1,p_2,\cdots,p_d)$ be the eigenvalues of $\rho$ arranged in increasing order. Further assume, $\vec \lambda = (\lambda_1,\lambda_2,\cdots,\lambda_d)$ be the eigenvalues of the Hermitian operator $\mathcal{O}$ arranged in decreasing order. It follows from Lemma \ref{lem3} 
\begin{align}
\mathcal{O}^p(\rho)=\vec \lambda \cdot \vec p = \sum_{i\in[d]}\lambda_ip_i.    
\end{align}
Consider now, $\vec q=(q_1,q_2,\cdots,q_d)$ to be the eigenvalues of $N(\rho)$ arranged in increasing order. Therefore, we have
\begin{align}\label{lem5eq1}
\mathcal{O}^p(N(\rho))-\mathcal{O}^p(\rho)&= \sum_{i\in [d]} \lambda_i(q_i-p_i) \nonumber \\
&= \left(q_1-p_1\right)\lambda_1+\Big((q_1+q_2)-(p_1+p_2)\Big)\lambda_2+\cdots+\left(\sum_{i\in[d]}q_i-\sum_{i\in[d]}p_i\right)\lambda_d \nonumber \\
&~~-\left\lbrace (q_1-p_1)\lambda_2+\Big((q_1+q_2)-(p_1+p_2)\Big)\lambda_3+\cdots+\left(\sum_{i\in[d-1]}q_i-\sum_{i\in[d-1]}p_i\right)\lambda_d\right\rbrace\nonumber\\
&=\left(q_1-p_1\right)\left(\lambda_1-\lambda_2\right)+\Big((q_1+q_2)-(p_1+p_2)\Big)\left(\lambda_2-\lambda_3\right)+\cdots\nonumber\\
&~~~~~~~~~~~+\left(\sum_{i\in[d-1]}q_i-\sum_{i\in[d-1]}p_i\right)\left(\lambda_{d-1}-\lambda_d\right)+\left(\sum_{i\in[d]}q_i-\sum_{i\in[d]}p_i\right)\lambda_d\nonumber \\
& = \sum_{l=2}^d\left\{\sum_{i\in[l-1]}\left(q_i-p_i\right)\right\}\Big(\lambda_{l-1}-\lambda_{l}\Big)+\left(\sum_{i\in[d]}(q_i-p_i)\right)\lambda_d\le 0.
\end{align}
Inequality in last line follows from the fact that $N$ is a unital and hence $\vec p$ majorizes $\vec q$, and the fact that $\vec \lambda$ is arranged in decreasing order. 

Now, $\mathcal{O}$ being non-degenerate we have $\lambda_{l-1}-\lambda_{l}> 0,~\forall~l$. Therefore, $\mathcal{O}^p(\rho)=\mathcal{O}^p(N(\rho))$ implies $\vec p =\vec q$. This concludes the proof. 
\end{proof}
\subsection*{Proof of Theorem \textbf{$2$}}
\begin{proof}
(Theorem \textbf{$2$}).\\ 
{\bf if part:} We first show that for every CQ state $\rho_{\bf A}$ of  Eq.(\ref{prop1eq1}), $\mathbb{O}_{G\mathcal{N}_C^0}(\rho_{\bf A})=0$. Notice, for every such $\rho_{\bf A}$, one can always find a unitary $u_1^\star\in\mathrm{U}_1$, such that:
\begin{align}
\sigma_{\bf A}=\Big(u_1^\star \otimes \mathbf{I}_2\Big)\rho_{\bf A}\Big(\left(u_1^\star\right)^\dag\otimes\mathbf{I}_2\Big)=\sum p_{i_1i_2}\ket{o^{i_1}}_1\bra{o^{i_1}} \otimes \psi^{i_2}_{2}.
\end{align}
Clearly, after applying $u_1^\star$ if the channel $N$ as defined in Eq.(\ref{l2b}) is used to send the $A_1$ sub-part to $A_2$, then the state remains unchanged, i.e., $(N\otimes\mathcal{I})\sigma_{\bf A}=\sigma_{\bf A}$. The states $\sigma_{\bf A}$ and $\rho_{\bf A}$ having same spectral yield $\mathcal{O}_{\bf A}^p(\sigma_{\bf A})=\mathcal{O}_{\bf A}^p(\rho_{\bf A})$.  Therefore, under the choice $\{\mathrm{U}_1,\mathrm{U}_2,\mathrm{U}_{12}\} = \{u_1^\star,\mathbf{I}_2,u_{12}^\star\}$, with  $u_{12}^\star$ mapping $\sigma_{\bf A}$ to the passive state having the same spectrum, one have
\begin{align*}
\Delta \mathcal{O}^{\{u_1^\star,\mathbf{I}_2,u_{12}^\star\}}_{N}(\rho_{\bf A}) &= \Delta \mathcal{O}_1^{u_1^\star}(\rho_1)+\Delta \mathcal{O}_2^{\mathbf{I}_2}(\rho_2)+\Delta \mathcal{O}_{\bf A}^{u_{12}^\star}(\sigma_{\bf A})\\
&=\left(\langle\mathcal{O}_1\rangle_{\rho_1}-\langle\mathcal{O}_1\rangle_{\sigma_1}\right)+\left(\langle\mathcal{O}_2\rangle_{\rho_2}-\langle\mathcal{O}_2\rangle_{\sigma_2}\right)+\Delta \mathcal{O}_{\bf A}^{u_{12}^\star}(\sigma_{\bf A})~~\mbox{[Since $\rho_2=\sigma_2$]}\\
&=\left(\langle\mathcal{O}_1\rangle_{\rho_1}+\langle\mathcal{O}_2\rangle_{\rho_2}\right)-\left(\langle\mathcal{O}_1\rangle_{\sigma_1}+\langle\mathcal{O}_2\rangle_{\sigma_2}\right)+\Delta \mathcal{O}_{\bf A}^{u_{12}^\star}(\sigma_{\bf A})\\
&= \Big(\langle\mathcal{O}_{\bf A}\rangle_{\rho_{\bf A}}-\langle\mathcal{O}_{\bf A}\rangle_{\sigma_{\bf A}}\Big)+\left(\langle\mathcal{O}_{\bf A}\rangle_{\sigma_{\bf A}}-\mathcal{O}_{\bf A}^p\left(\rho_{\bf A}\right)\right)~~\mbox{[using Lemma \ref{lem4}]}\\
& =\langle\mathcal{O}_{\bf A}\rangle_{\rho_{\bf A}}-\mathcal{O}_{\bf A}^p\left(\rho_{\bf A}\right) = \Delta \mathcal{O}_{\bf A}^{\mathrm{U}_{\bf A}}(\rho_{\bf A}) ~~\mbox{[using Lemma \ref{lem4}]}.
\end{align*}
Hence, the choice $\{u_1^\star,\mathbf{I}_2,u_{12}^\star\}$ being optimal yields $\mathbb{O}_{G\mathcal{N}^0_C}(\rho_{\bf A})=0$.\\\\
{\bf only if part:} Given a state $\chi_{\bf A}\in\mathcal{D}(\mathbb{C}^{d_1}\otimes\mathbb{C}^{d_2})$, let $\{v_1,v_2,v_{12}\}$ be the optimal choice for evaluating $\mathbb{O}_{G\mathcal{N}^0_C}(\chi_{\bf A})$.
Furthermore, let us assume, $\eta_{\bf A} = (v_1\otimes v_2)\chi_{\bf A}(v_1^\dag \otimes v_2^\dag)$. Thereafter when $A_1$ sub-part is send through the channel $N$ of the form Eq.(\ref{l2b}), the state becomes $(N\otimes\mathcal{I})\eta_{\bf A}=\tilde \eta_{\bf A}$ (say). Thus we have,
\begin{align}
\Delta \mathcal{O}^{\{\mathrm{U}_1,\mathrm{U}_2,\mathrm{U}_{12}\}}_{\mathcal{N}^0_C}(\chi_{\bf A})&=\Delta \mathcal{O}^{\{v_1,v_2,v_{12}\}}_{N}(\chi_{\bf A})\nonumber\\
&= \Big(\langle\mathcal{O}_{\bf A}\rangle_{\chi_{\bf A}}-\langle\mathcal{O}_{\bf A}\rangle_{\eta_{\bf A}}\Big)+\left(\langle\mathcal{O}_{\bf A}\rangle_{\tilde \eta_{\bf A}}-\mathcal{O}_{\bf A}^p\left(\tilde \eta_{\bf A}\right)\right)\nonumber\\
&=\langle\mathcal{O}_{\bf A}\rangle_{\chi_{\bf A}}-\mathcal{O}_{\bf A}^p\left(\tilde \eta_{\bf A}\right).\label{oif1}
\end{align}
In the last line we use the fact that $\mathcal{O}_{\bf A}(\zeta_{\bf A}) = \mathcal{O}_{\bf A}((N\otimes\mathcal{I})\zeta_{\bf A}),~\forall~ \zeta_{\bf A}\in\mathcal{D}(\mathbb{C}^{d_1}\otimes\mathbb{C}^{d_2})$. Now, using Eq.(\ref{oif1}) we obtain
\begin{align}\label{del}
\mathbb{O}_{G\mathcal{N}^0_C}(\chi_{\bf A})&=\Delta \mathcal{O_{\bf A}}(\chi_{A}) - \Delta \mathcal{O}^{\{\mathrm{U}_1,\mathrm{U}_2,\mathrm{U}_{12}\}}_{\mathcal{N}^0_C}(\chi_{\bf A}) \nonumber \\
&=\left(\mathcal{O}_{\bf A}\left(\chi_{\bf A}\right)-\mathcal{O}^p_{\bf A}\left(\chi_{\bf A}\right)\right)-\left(\mathcal{O}_{\bf A}\left(\chi_{\bf A}\right)-\mathcal{O}_{\bf A}^p\left(\tilde \eta_{\bf A}\right)\right)\nonumber\\
&=\mathcal{O}_{\bf A}^p\left(\tilde \eta_{\bf A}\right)- \mathcal{O}_{\bf A}^p\left(\chi_{\bf A}\right)\nonumber\\
&=\mathcal{O}_{\bf A}^p\left(\tilde \eta_{\bf A}\right)- \mathcal{O}_{\bf A}^p\left(\eta_{\bf A}\right)~~\mbox{[Since $\mathcal{O}^p(\rho)=\mathcal{O}^p(u\rho u^\dagger)$]}.
\end{align}
Therefore, $ \mathbb{O}_{G\mathcal{N}^0_C}(\chi_{\bf A})=0$ implies $\mathcal{O}_{\bf A}^p\left(\tilde \eta_{\bf A}\right)= \mathcal{O}_{\bf A}^p\left(\eta_{\bf A}\right)$. Using Lemma \ref{lem5}, we can further claim that, $\eta_{\bf A}$ is a fixed point of $N$, i.e,
\begin{align}
\eta_{\bf A}&= \left(N\otimes\mathcal{I}\right)\eta_{\bf A}\nonumber\\
&=\sum_{l\in [d]}\tr\left(\left(\ket{o^l}_1\bra{o^l}\right)\eta_{\bf A}\right)\ket{o^l}_1\bra{o^l}\otimes\eta_{A_2|l}\nonumber\\
\Rightarrow (v_1\otimes v_2)\chi_{\bf A}(v_1^\dag \otimes v_2^\dag)&= \sum_{l\in [d]}\tr\left(\left(\ket{o^l}_1\bra{o^l}\right)\eta_{\bf A}\right)\ket{o^l}_1\bra{o^l}\otimes\eta_{A_2|l}\nonumber\\
\Rightarrow \chi_{\bf A} &= (v_1^\dag \otimes v_2^\dag)\underbrace{\left(\sum_{l\in [d]}\tr\left(\left(\ket{o^l}_1\bra{o^l}\right)\eta_{\bf A}\right)\ket{o^l}_1\bra{o^l}\otimes\eta_{A_2|l}\right)}_{\mbox{CQ state}}(v_1\otimes v_2),
\end{align}
where $\eta_{A_2|l}:=\tr_{1}\left[\left(\ket{o^l}_1\bra{o^l}\right)\eta_{\bf A}\right]/\tr\left[\left(\ket{o^l}_1\bra{o^l}\right)\eta_{\bf A}\right]$ represents the normalized state of subpart $A_2$, conditioned that the projector $\ket{o^l}_1\bra{o^l}$ acts on part $A_1$. Since the set of CQ states is invariant under local unitary operations and $\eta_{\bf A}$ is a CQ state, therefore $\chi_{\bf A}$ is also a CQ state. This completes the proof.
\end{proof}




\twocolumngrid

\bibliography{Bibliography}

\end{document}